\documentclass[a4paper,10pt]{article}
\usepackage{amsthm}
\usepackage{graphicx,amssymb}
\usepackage{epsfig, color, graphicx}

\newtheorem{lemma}{Lemma}
\newtheorem{nc}{Necessary Condition}
\newtheorem{nc'}{Necessary Condition}
\newtheorem{corollary}{Corollary}
\newtheorem{theorem}{Theorem}
\usepackage{graphicx}
\usepackage{epstopdf}
\usepackage{hyperref}
\usepackage{bbm}
\usepackage[ngerman,american]{babel}           
\usepackage[square,numbers]{natbib} 
\usepackage[fleqn]{amsmath}
\usepackage{makerobust} 
\usepackage[margin=1.0in]{geometry}
\usepackage{enumerate}

\usepackage{epsfig, color, graphicx}
\usepackage{graphicx}
\usepackage{epstopdf}
\usepackage{textcomp}

\begin{document}
\centerline{\LARGE{\bf Some Results On Point Visibility Graphs  \footnote{An extended abstract of this paper
appeared in the proceedings of the Eighth International Workshop on Algorithms and Computation, pp.163-175, 2014 \cite{recogpvg-2014}.}}}

\vskip 0.7 in
\begin{center}
\mbox{\begin{minipage} [b] {2.5in}
\centerline{Subir Kumar Ghosh
}
\centerline{School of Technology \& Computer Science}
\centerline{Tata Institute of Fundamental Research}
\centerline{Mumbai 400005, India}
\centerline{ghosh@tifr.res.in}
\end{minipage}}\hspace{0.8in}
\mbox{\begin{minipage} [b] {2.5in}
\centerline{Bodhayan Roy}
\centerline{ School of Technology \& Computer Science}
\centerline{ Tata Institute of Fundamental Research}
\centerline{Mumbai 400005, India}
\centerline{bodhayan@tifr.res.in}
\end{minipage}}

\end{center}

\vskip 0.3 in


\begin{abstract}
In this paper, we present three necessary conditions for recognizing   point visibility graphs.
We show that this recognition problem lies in PSPACE.
We state new properties of point visibility graphs along with some known properties that are important
in understanding point visibility graphs.
For planar point visibility graphs, we present a 
complete characterization which leads to a linear time recognition and reconstruction algorithm. 
\end{abstract}

\section{Introduction}
\noindent The visibility graph is a fundamental structure studied in the field of computational geometry  
and geometric graph theory \cite{bcko-cgaa-08, g-vap-07}. 
 Some of the early applications of visibility graphs included computing 
Euclidean shortest paths in the presence of obstacles \cite{lw-apcf-79} and decomposing 
two-dimensional shapes into clusters \cite{sh-ddsg-79}. 
Here, we consider problems from visibility graph theory.
$\\ \\$
Let $P=\{ p_1, p_2, ..., p_n \}$ be a set of points in the plane (see Fig. \ref{visgr}). 
We say that two points $p_i$ and $p_j$ of $P$ 
are \emph{mutually visible} if the line 
segment $p_ip_j$ does not contain or pass through any other point of $P$. In other words,
$p_i$  and $p_j$ are visible if $P \cap  {p_ip_j}=\{p_i, p_j\}$. If two vertices are not visible, 
they are called an {\it invisible pair}. For example, 
in Fig. \ref{visgr}(c), $p_1$ and $p_5$ form a visible pair whereas $p_1$ and $p_3$ form an invisible pair.
If a point $p_k \in P$ lies on the segment $p_ip_j$ connecting two points $p_i$ and $p_j$ in $P$, 
we say that $p_k$ blocks the visibility between $p_i$ and $p_j$, and
$p_k$ is called a {\it blocker} in $P$.  For example in Fig. \ref{visgr}(c), 
$p_5$ blocks the visibility between $p_1$ and $p_3$ as $p_5$ lies on the segment $p_1p_3$.
The {\it visibility graph} 
(also called the {\it point visibility graph} (PVG))
$G$ of $P$ is defined 
by associating a vertex $v_i$ with each point $p_i$ of $P$ 
such that $(v_i, v_j)$ is an undirected edge of $G$ if and only if $p_i$ and $p_j$
are mutually visible (see Fig. \ref{visgr}(a)). Observe that if no three points of 
$P$ are collinear, then $G$ is a complete graph as
each pair of points in $P$ is visible since there is no blocker in $P$. 
Sometimes the visibility graph is drawn directly on the point set, 
as shown in Figs. \ref{visgr}(b) and \ref{visgr}(c), which is referred to as a {\it visibility embedding} of $G$.
$\\ \\$
Given a point set $P$, the visibility  graph $G$ of $P$ can be computed as follows.
 For each point
$p_i$ of $P$, the points of $P$ are sorted in angular order around $p_i$.  If two points $p_j$ and $p_k$ 
are consecutive in the sorted order, check whether $p_i$, $p_j$ and $p_k$ are collinear points.
By traversing the sorted order, all points of $P$, that are not visible from $p_i$, can be identified 
in $O(n \log n)$ time. Hence, $G$ can be computed from $P$ in $O(n ^2 \log n)$ time.
Using the result of Chazelle et al. \cite{cgl-pgd-85} or Edelsbrunner et al. 
\cite{Edelsbrunner:1986:CAL}, the  time complexity of the algorithm can be improved to $O(n^2)$  
by computing sorted angular orders for all points together 
in $O(n^2)$ time.
$\\ \\$
 Consider the opposite problem of determining if there is a set of points $P$ 
whose visibility graph is the given graph $G$. This problem is called the visibility graph 
{\it recognition} problem.  Identifying the set of properties satisfied by all visibility 
graphs is called the visibility graph {\it characterization} problem. The problem of 
actually drawing one such set of points $P$ whose visibility graph is the given graph $G$, 
is called the visibility graph {\it reconstruction} problem.
$\\ \\$
Here we consider the recognition
problem: Given a graph $G$ in adjacency matrix form, determine whether $G$ is the visibility graph of a set
of points $P$ in the plane \cite{prob-ghosh}. 
In Sect. 2, we present three necessary conditions for this recognition problem.
In the same section, we establish new properties of point visibility graphs, and in addition,
we state some known properties with proofs that are important in understanding point visibility graphs.
 Though the first necessary condition can be tested in $O(n^3)$ time,
it is not clear whether the second necessary and third conditions can be tested in polynomial time.
On the other hand, we show in Sect. 3 that the recognition problem lies in PSPACE. 
$\\ \\$
If a given graph $G$ is  planar,  there can be three cases: (i) $G$   has a planar visibility embedding (Fig. \ref{plpvg1}), 
(ii) $G$ admits a visibility embedding, but no visibility embedding of $G$ is planar (Fig. \ref{plpvg2}), and
(iii) $G$  does not have any visibility embedding (Fig. \ref{plpvg3}). 
Case (i) has been characterized by Eppstein 
\cite{epp-plpvg} by presenting four infinite families of $G$ and one particular graph. 
In order to characterize graphs in Case (i) and Case (ii), we show that two infinite families and five particular graphs are 
required in addition to graphs for Case (i).
 Using this characterization, we present an $O(n)$ algorithm for recognizing and reconstructing $G$ in Sect. 4.
Note that this algorithm does not require any prior embedding of $G$. Finally,
we conclude the paper with a few remarks.
\begin{figure} [h]
\begin{center}
\mbox{\begin{minipage} [b] {90mm}
\centerline{\hbox{\psfig{figure=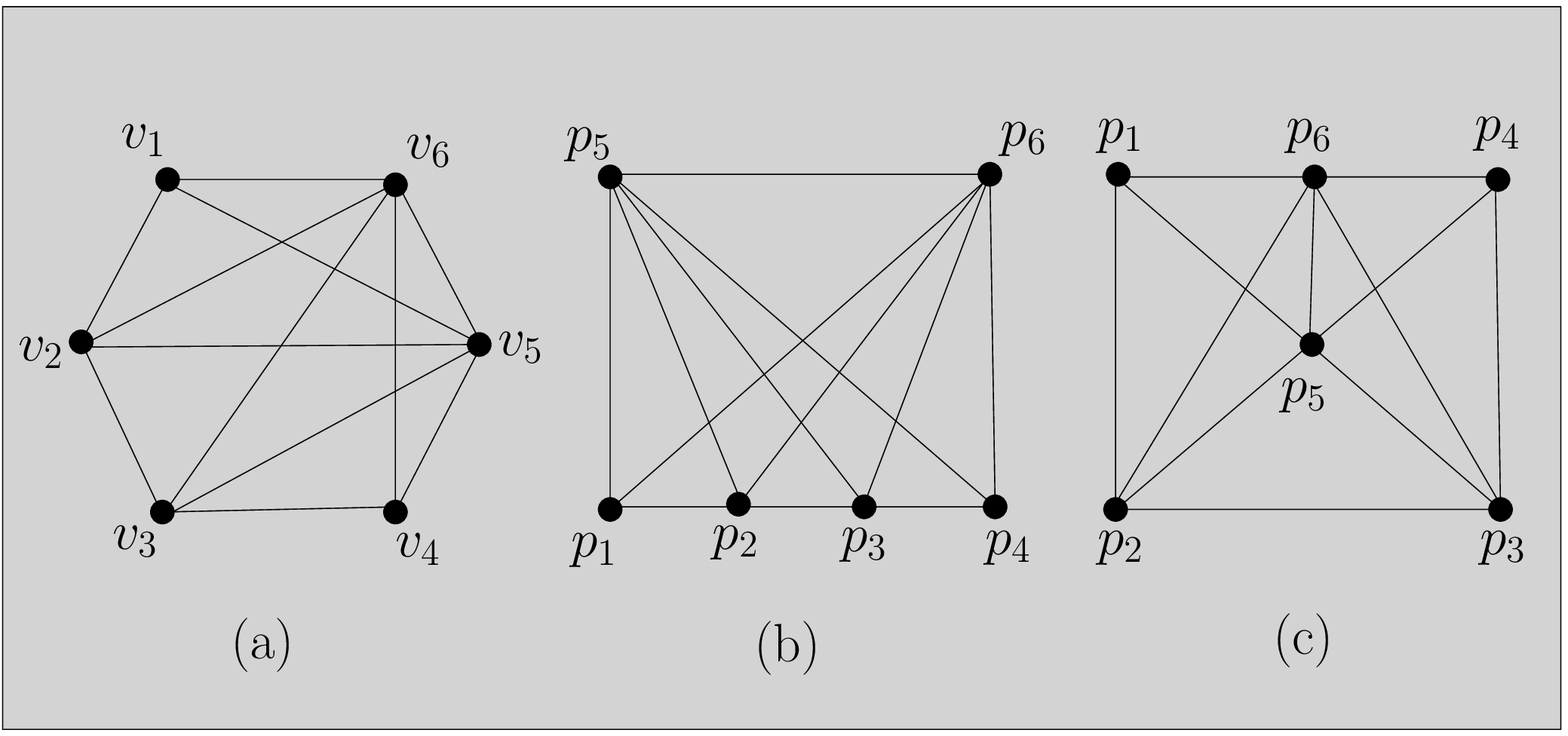,width=0.99\hsize}}}
\caption{ (a) A point visibility graph with $(v_1,v_2,v_3,v_4)$ as a CSP. (b) A visibility embedding of the point visibility graph
 where  $(p_1,p_2,p_3,p_4)$ is a GSP. 
                 (c) A visibility embedding of the point visibility graph where  $(p_1,p_2,p_3,p_4)$ is not a GSP. }
\label{visgr}
\end{minipage}}\hspace{2mm}
\mbox{\begin{minipage} [b] {48mm}
\centerline{\hbox{\psfig{figure=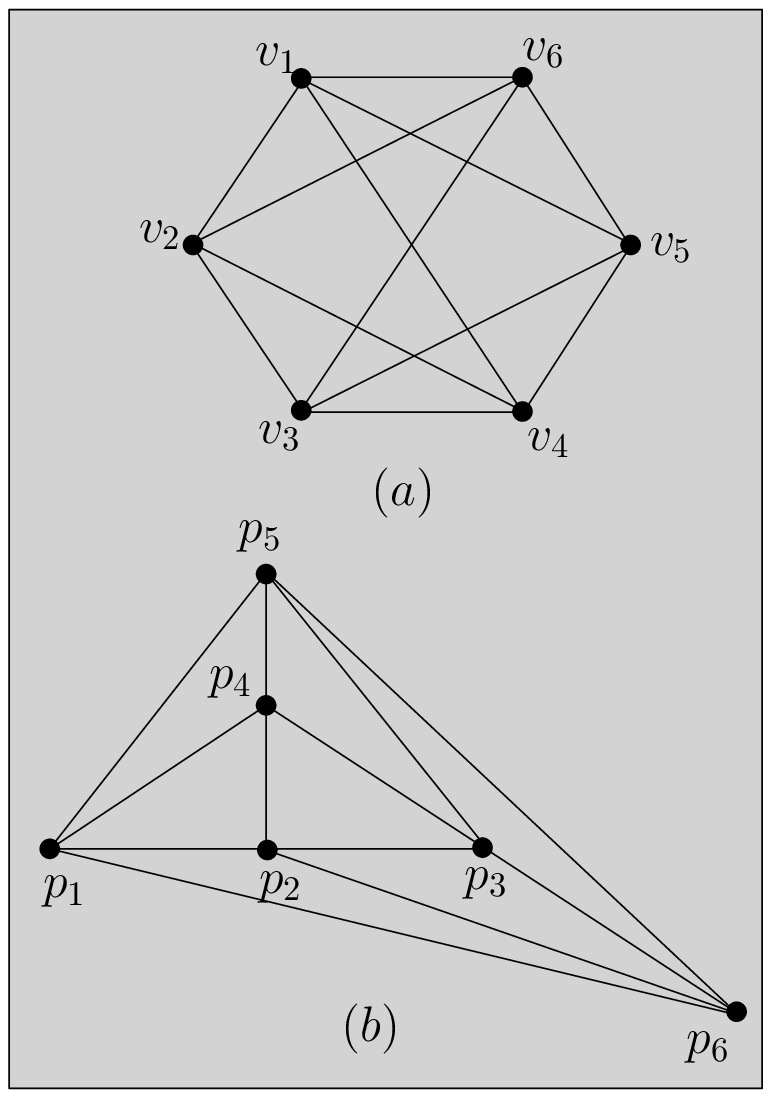,width=0.99\hsize}}}
\caption{ (a) A planar graph $G$.
  (b) A planar visibility embedding of $G$.}
\label{plpvg1}
\end{minipage}}
\end{center}
\end{figure}

\begin{figure} [h]
\begin{center}
\mbox{\begin{minipage} [b] {90mm}
\centerline{\hbox{\psfig{figure=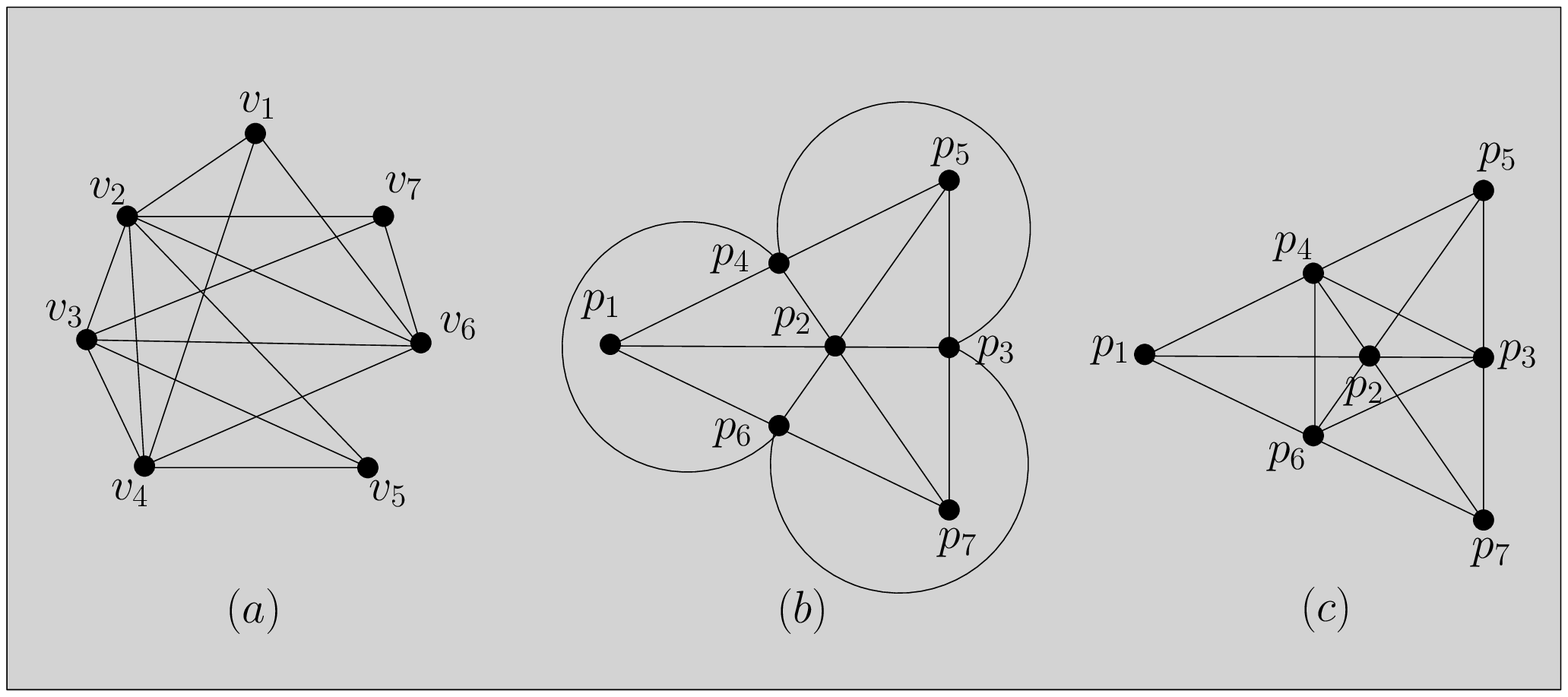,width=0.95\hsize}}}
\caption{ (a) A planar graph $G$.
  (b) A planar embedding of $G$.  (c) A non-planar visibility embedding of $G$}
 \label{plpvg2}
\end{minipage}}\hspace{02mm}
\mbox{\begin{minipage} [b] {48mm}
\centerline{\hbox{\psfig{figure=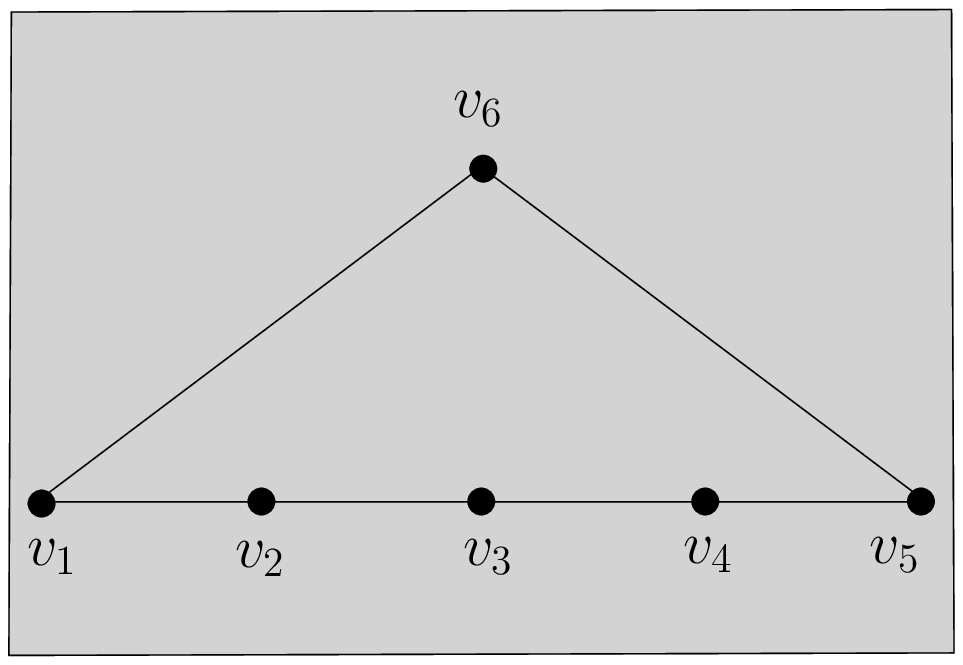,width=0.95\hsize}}}
\caption{   A planar graph $G$ that does not admit a visibility embedding.
    }
 \label{plpvg3}
\end{minipage}}
\end{center}
\end{figure}

\section{Properties of point visibility graphs }
Consider a subset $S$ of vertices of $G$ such that their corresponding points C in a visibility embedding $\xi$ of $G$ are 
collinear.  The path formed by the points of C is called a {\it geometric
straight path} (GSP). 
For example,  the path  $(p_1,p_2,p_3,p_4)$ in Fig. \ref{visgr}(b) is a GSP as the points $p_1$,
$p_2$, $p_3$ and $p_4$ are collinear. Note that there may 
be another visibility embedding $\xi$
of $G$ as shown in Fig. \ref{visgr}(c), where points $p_1$, $p_2$, $p_3$ and $p_4$ are not collinear.
So, the points forming a GSP in $\xi$ may not form a GSP in every visibility embedding of $G$. 
If a GSP is a maximal set of collinear points, it is called a {\it maximal geometric
straight path} (max GSP).
A GSP of $k$ collinear points is denoted as \emph{k-GSP}.
 In the following, we state some properties of PVGs and present three necessary conditions for recognizing $G$.
\begin{lemma} \label{adj} If $G$ is a PVG but not a path, 
then for any GSP in any visibility embedding of $G$, there is a point
visible from all the points of the GSP\cite{kpw-ocnv-2005}.
 \end{lemma}
\begin{proof} For every GSP, there exists a point $p_i$ whose perpendicular distance to the line containing the 
GSP is the smallest. So, all points of the GSP are visible from $p_i$.   \end{proof}
\begin{lemma} \label{deg}
If $G$ admits a visibility embedding $\xi$ having a $k$-GSP, then the number of edges in $G$ is 
at least $(k-1) + k(n-k)$.
\end{lemma}
%
%
%
%
\begin{proof} Let $p_i$ and $p_j$ be two points of $\xi$ such that $p_i$ is a point of the $k$-GSP and $p_j$ is not. 
Consider the segment $p_ip_j$. If $p_i$ and $p_j$ are mutually visible, then $(v_i,v_j)$ is an edge in $G$.
Otherwise, there exists a blocker $p_k$  on $p_ip_j$ such that $(v_j,v_k)$ is an edge in $G$.
So, $p_j$ has an edge in the direction towards $p_i$. Therefore, for every such pair $p_i$ and $p_j$, there is
an edge in $G$. So, $(n-k)k$ such pairs in $\xi$ correspond to $(n-k)k$ edges in $G$.
Moreover, there are $(k-1)$ edges in $G$ corresponding to the $k$-GSP.
Hence, $G$ has at least  $(k-1) + k(n-k)$ edges.  \end{proof}
\begin{corollary}
If a point $p_i$ in a visibility embedding $\xi$ of $G$ does not belong to a $k$-GSP in $\xi$,
 then its corresponding vertex $v_i$ in $G$ has degree at least $k$.
\end{corollary}

$\\ $
Let $H$ be a path in $G$ such that no edges exist between any two non-consecutive vertices in $H$.
We call $H$ a 
{\it combinatorial straight path} $(CSP)$. Observe that in a visibility embedding of $G$, $H$
may not always correspond to a GSP. In Fig. \ref{visgr}(a),  $H$ =  $(v_1,v_2,v_3,v_4)$ 
is a CSP which corresponds to a GSP in
Fig.  \ref{visgr}(b) but not in Fig.  \ref{visgr}(c). Note that a CSP always refers to a path in $G$, whereas a GSP refers to
a path in a visibility embedding of $G$. A CSP that is 
a maximal path, is called a {\it maximal combinatorial
straight path} $(max \ CSP)$.
A CSP of $k$-vertices is denoted as \emph{k-CSP}.
\begin{lemma}\label{bip} 
$G$ is a PVG and bipartite if and only if the entire $G$ is a CSP.
\end{lemma}
\begin{proof} If the entire $G$ can be embedded as a GSP, then alternating points in the GSP form the bipartition and the lemma holds. Otherwise,
there exists at least one max GSP which does not contain all the points.
By Lemma \ref{adj}, there exists one  point $p_i$ adjacent to all points of the GSP. So, $p_i$ must belong to one partition and all   
points of the GSP (having edges) belong to the other partition. Hence, $G$ cannot be a bipartite graph, a contradiction. 
The other direction of the proof is trivial.  \end{proof}
\begin{corollary}
 $G$ is a PVG and triangle-free if and only if the entire $G$ is a CSP.
\end{corollary}
\begin{lemma} \label{maxcl}
If $G$ is a PVG, then the size of the maximum clique in $G$ is bounded by twice the minimum degree of $G$, and the bound is tight.
\end{lemma}
\begin{proof} 
In a visibility embedding of $G$, draw rays from a point $p_i$ of minimum degree through every visible point of $p_i$.
Observe that any ray may contain several points not visible from $p_i$. 
Since any clique can have at most two points from the same ray, 
the size of the clique is at most twice the number of rays, which gives twice the minimum degree of $G$.    \end{proof}
\begin{lemma} \label{diam}If $G$ is a PVG and it has more than one max CSP, 
then the diameter of $G$ is 2 \cite{kpw-ocnv-2005}. \end{lemma}
\begin{proof} If two vertices $v_i$ and  $v_j$ are not adjacent in $G$,
 then they belong to a CSP L of length at least two. By Lemma \ref{adj},
  there must be some vertex $v_k$ that is adjacent  to every vertex in L. $(v_i, v_k, v_j)$ is the required path
of length 2. Therefore, the diameter of $G$ cannot be more than two.   \end{proof} 
\begin{corollary} \label{bfs}
 If $G$ is a PVG but not a path, then the BFS tree of $G$ rooted at any  vertex $v_i$ of
G has at most three levels consisting of $v_i$ in the first level, the neighbours of $v_i$ in $G$ 
in the second level, and the rest of the vertices of $G$ in the third level. 
\end{corollary}
\begin{lemma} \label{conn}
If $G$ is a PVG but not a path, then the subgraph induced by the neighbours of any vertex $v_i$, excluding $v_i$, is connected.\end{lemma}
\begin{proof} Consider a visibility embedding of $G$ where $G$ is not a path. Let $(u_1,u_2,...,u_k,u_1)$ be the visible points
of $p_i$ in clockwise angular order. If $p_i$ is not a convex hull point, then $(u_1,u_2),(u_2,u_3),
...,(u_{k-1},u_k),$ $(u_k,u_1)$ are visible pairs (Fig. \ref{convh}(a)). 
If $p_i$, $u_1$ and $u_k$ are convex hull points, then $(u_1,u_2),(u_2,u_3),
...,$ $(u_{k-1},u_k)$ are visible pairs (Fig. \ref{convh}(b)). 
Since there exists a path between every pair of points in $(u_1,u_2,...,u_k,u_1)$,
 the subgraph induced by the neighbours of $v_i$ is connected.   \end{proof}
%
\begin{figure} 
\begin{center}
\centerline{\hbox{\psfig{figure=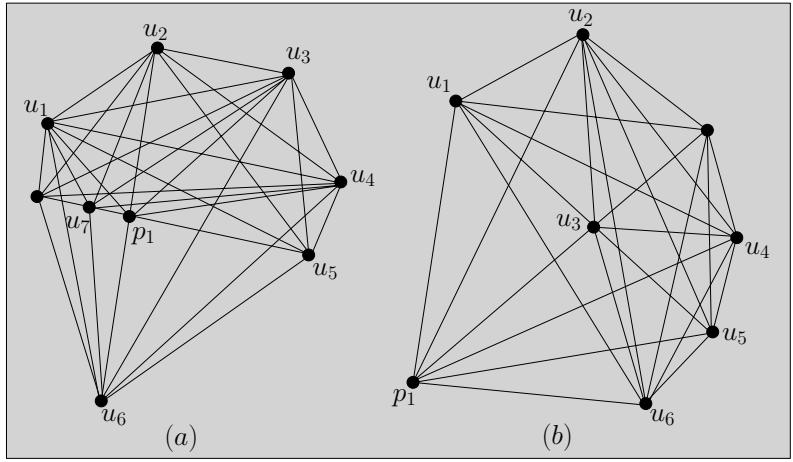,width=0.65\hsize}}}
\caption{ (a) The points $(u_1, u_2, ..., u_7$, $u_1)$ are visible from an internal point $p_1$.
  (b) The points $(u_1,u_2,...,u_6)$ are visible from a convex hull point $p_1$.  }
\label{convh}
\end{center}
\end{figure}
%
%
\begin{nc} \label {nc1}
 If $G$ is not a CSP, then the BFS tree of $G$ rooted at any vertex can have at most three levels,
and the induced subgraph formed by the vertices in the second level  must be connected.
\end{nc} 
\begin{proof} Follows from Corollary \ref{bfs} and Lemma \ref{conn}.   \end{proof}
$\\ \\$
As defined for point sets, if two vertices $v_i$ and $v_j$ of $G$ are adjacent (or, not adjacent) in $G$, $(v_i,v_j)$ is referred to as
a {\it visible pair} (respectively, {\it invisible pair}) of $G$.
Let $(v_1,v_2,...,v_k)$ be a path in $G$ such that no two non-consecutive vertices are connected by an edge in $G$ 
(Fig.  \ref{ncdiag}(a)). 
For any vertex $v_j$, $2 \leq j \leq k-1$, $v_j$ is called a \emph{vertex-blocker} of $(v_{j-1},v_{j+1})$ as 
 $(v_{j-1},v_{j+1})$ is not an edge in $G$ and both  $(v_{j-1},v_{j})$ and  $(v_{j},v_{j+1})$ are 
edges in $G$. In the same way, consecutive vertex-blockers on such a path are also called \emph{vertex-blockers}. For example,
$v_m*v_{m+1}$ is a vertex-blocker of $(v_{m-1},v_{m+2})$ for $2 \leq m \leq k-2$.
Note that $*$ represents concatenation of consecutive vertex-blockers.
$\\ \\$
Consider the graph in Fig. \ref{ncdiag}(b). Though $G$ satisfies Necessary Condition 1, it is not a PVG because it does not admit a 
visibility embedding. 
It can be seen that this graph without the edge $(v_2,v_4)$ admits a visibility embedding (see Fig. \ref{ncdiag}(a)),
 where $(v_1,v_2,v_3,v_4,v_5)$ forms a GSP.
However, $(v_2,v_4)$ demands visibility between two non-consecutive collinear blockers which cannot be realized in any visibility embedding.
\begin{figure} 
\begin{center}
\centerline{\hbox{\psfig{figure=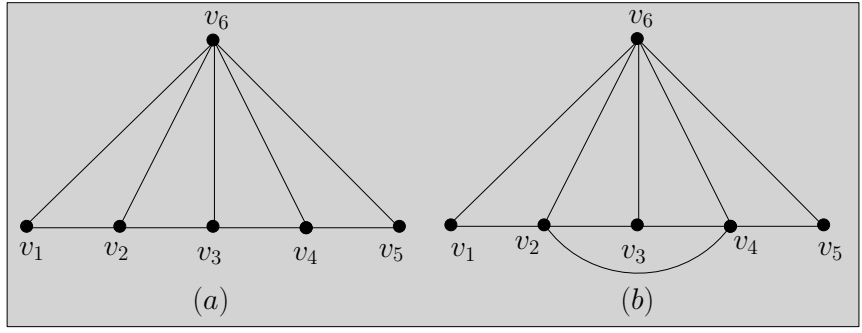,width=0.705\hsize}}}
\caption{ (a) Vertices $v_2$, $v_3$, $v_4$ are vertex-blockers of $(v_1,v_3)$, $(v_3,v_4)$ $(v_3,v_5)$ respectively.
Also, $v_2*v_3*v_4$ is the vertex-blocker of $(v_1,v_5)$.  (b) The graph satisfies Necessary
Condition 1 but is not a PVG because of the edge $(v_2,v_4)$.}
\label{ncdiag}
\end{center}
\end{figure}

 \begin{nc}\label {nc2}
  There exists an assignment of vertex-blockers to invisible pairs in $G$ such that:
\begin{enumerate}
 \item Every invisible pair is assigned one vertex-blocker.
 \item If two invisible pairs  in $G$ sharing a vertex $v_i$ $($say, $(v_i,v_j)$ and $(v_i,v_k)$ $)$, and 
their assigned vertex-blockers are not disjoint, then all vertices in the two assigned vertex-blockers  
along with vertices $v_i$, $v_j$ and $v_k$ must be a CSP in $G$.
\item If two invisible pairs in $G$ are sharing a vertex $v_i$ (say, $(v_i, v_j)$ and $(v_i, v_k)$),
and $v_k$ is assigned as a vertex blocker to $(v_i, v_j)$, then $v_j$ is not assigned as a vertex blocker to $(v_i, v_k)$.
\end{enumerate}
\end{nc}
\begin{proof} In a visibility embedding of $G$, every segment connecting two points, that are not mutually visible,
must pass through another point or a set of collinear points, and they correspond to vertex-blockers in $G$. 
$\\ \\$
Since $(v_i,v_j)$ and $(v_i,v_k)$ are invisible pairs, the segments $(p_i,p_j)$ and  $(p_i,p_k)$  must contain points. 
 If there exists a point $p_m$
on both $p_ip_j$ and $p_ip_k$, then points $p_i$, $p_m$, $p_j$, $p_k$ must be collinear. So, $v_i$, $v_m$, $v_j$ and $v_k$ 
must belong to a CSP.  
$\\ \\$
Since $(v_i, v_j)$ and $(v_i, v_k)$ are invisible pairs, the segments $(p_i, p_j)$ and $(p_i, p_k)$ must contain
points. If the point $p_k$ lies on $p_ip_j$, then $p_j$ cannot lie on $p_ip_k$, because it contradicts 
the order of points on a line.   \end{proof}
$\\ \\$
Consider the graph $G$ in Fig. \ref{figncprob}(a).
From its visibility embedding, it is clear that $G$ is a PVG and therefore, 
satisfies both Necessary Conditions \ref{nc1} and \ref{nc2}.
Let us construct a new graph $G'$ from $G$ by replacing 
edges $v_9v_{10}$ and $v_{11}v_{12}$
of $G$ by 
$v_9v_{11}$ and $v_{10}v_{12}$ (see Fig. \ref{figncprob}(b)). 
We have the following lemmas on $G'$.
\begin{lemma}
The graph $G'$ satisfies Necessary Conditions \ref{nc1} and \ref{nc2}.
\end{lemma}
\begin{proof}
Observe that the neighbours of any vertex in $G'$ induce a connected subgraph. Also,
the diameter of $G'$ is still two. Therefore, $G'$ satisfies Necessary Condition \ref{nc1}.
$\\ \\$
For showing that $G'$ also satisfies Necessary Condition \ref{nc2},
we consider the assignment of blockers to the mutually invisible pairs of vertices in $G'$ as follows: 
{$(v_0,v_5) \longrightarrow v_1$},
  {$(v_0,v_9) \longrightarrow v_1 \ast v_5$},
  {$(v_1,v_9) \longrightarrow v_5$},
  {$(v_0,v_6) \longrightarrow v_2$},
  {$(v_0,v_{10}) \longrightarrow v_2 \ast v_6$},
  {$(v_2,v_{10}) \longrightarrow v_6$}, 
  {$(v_0,v_7) \longrightarrow v_3$}, 
  {$(v_0,v_{11}) \longrightarrow v_3 \ast v_7$}, 
 {$(v_3,v_{11}) \longrightarrow v_7$},  
 {$(v_0,v_8) \longrightarrow v_4$},  
  {$(v_0,v_{12}) \longrightarrow v_4 \ast v_8$}, 
  {$(v_4,v_{12}) \longrightarrow v_8$}, 
   {$(v_1,v_3) \longrightarrow v_2$},
   {$(v_1,v_4) \longrightarrow v_2 \ast v_3$},
   {$(v_2,v_4) \longrightarrow v_3$},
   {$(v_5,v_7) \longrightarrow v_6$},
   {$(v_5,v_8) \longrightarrow v_6 \ast v_7$},
   {$(v_6,v_8) \longrightarrow v_7$},
   {$(v_9,v_{10}) \longrightarrow v_{11}$},
  {$(v_9,v_{12}) \longrightarrow v_{11} \ast v_{10}$}, 
  {$(v_{11},v_{12}) \longrightarrow v_{10}$}. 
  Observe that since the invisible pairs $(v_9,v_{11})$ and $(v_{10},v_{12})$ in $G$
  are replaced by $(v_9,v_{10})$ and $(v_{11},v_{12})$ in $G'$, the vertex-blocker  
 assignments have changed accordingly.
  It can be seen that the above assignment of vertex blockers satisfies Necessary Condition \ref{nc2}.
  
\end{proof}

\begin{lemma}
 The graph $G'$ is not a PVG.
\end{lemma}
\begin{proof}
   Let us assume on the contrary that $G$
  has a visibility embedding (say, $\xi$).
  Let $p_0, p_1, \ldots, p_{12}$ be the points of $\xi$ corresponding to the vertices $v_0, v_1, \ldots, v_{12}$ 
  respectively.
  Consider the rays $\overrightarrow{p_0p_1}$, $\overrightarrow{p_0p_2}$,
   $\overrightarrow{p_0p_3}$ and $\overrightarrow{p_0p_4}$. Since $v_0$ is not adjacent 
   to any of $v_5, v_6, \ldots, v_{12}$ in $G'$, $p_5, p_6, \ldots, p_{12}$ must lie on these four rays.
   $\\ \\$
   Consider the case where $p_0$ is not a blocker of $(p_1,p_4)$. 
   So, the angle at $p_0$ between $\overrightarrow{p_0p_1}$ and $\overrightarrow{p_0p_4}$ is not $180^{\circ}$.
   Let $w_{1,4}$ denote the wedge formed by $\overrightarrow{p_0p_1}$ and $\overrightarrow{p_0p_4}$
   such that the internal angle of $w_{1,4}$ is convex. Since a blocker of $(p_1,p_4)$ must lie on  
   $\overrightarrow{p_0p_2}$ or $\overrightarrow{p_0p_3}$ (say, $\overrightarrow{p_0p_2}$), 
   $\overrightarrow{p_0p_2}$ divides $w_{1,4}$ into wedges $w_{1,2}$ and $w_{2,4}$
%
%
   By a similar argument for $(p_2,p_4)$, $\overrightarrow{p_0,p_3}$
   passes through $w_{2,4}$.
   So, the ordering of the rays around $p_0$ in $w_{1,4}$ is 
   $(\overrightarrow{p_0p_1},\overrightarrow{p_0p_2} ,
    \overrightarrow{p_0p_3},\overrightarrow{p_0p_4})$.
$\\ \\$
   Let us locate the positions of $p_5$, $p_6$, $p_7$ and $p_8$ on $\overrightarrow{p_0p_1}$, $\overrightarrow{p_0p_2}$,
   $\overrightarrow{p_0p_3}$ and $\overrightarrow{p_0p_4}$.
   Observe that since each of the vertices $v_5$, $v_6$, $v_7$ and $v_8$ are adjacent to all of the 
   vertices $v_1$, $v_2$, $v_3$ and $v_4$, the points $p_5$, $p_6$, $p_7$ and $p_8$ must be the next points on 
   $\overrightarrow{p_0p_1}$, $\overrightarrow{p_0p_2}$,
   $\overrightarrow{p_0p_3}$ and $\overrightarrow{p_0p_4}$.
   In fact, the only two possibilities are 
   $( \overrightarrow{p_0p_1p_5}, \overrightarrow{p_0p_2p_6},
   \overrightarrow{p_0p_3p_7}, \overrightarrow{p_0p_4p_8})$
   and 
     $(\overrightarrow{p_0p_1p_8}, \overrightarrow{p_0p_2p_7},
   \overrightarrow{p_0p_3p_6}, \overrightarrow{p_0p_4p_5})$
   that can satisfy the blocking requirements among $p_5$, $p_6$, $p_7$ and $p_8$.
   $\\ \\$
   Let us locate the positions of $p_9$, $p_{10}$, $p_{11}$ and $p_{12}$ on $\overrightarrow{p_0p_1}$, $\overrightarrow{p_0p_2}$,
   $\overrightarrow{p_0p_3}$ and $\overrightarrow{p_0p_4}$.
   Since 
   $v_9$ is adjacent to $v_2$, $v_3$ and $v_4$ but not to $v_1$ in $G'$, 
   $p_9$ must lie on $\overrightarrow{p_0p_1}$.
   Similarly, $p_{10}$ must lie on $\overrightarrow{p_0p_2}$. 
   Since $p_9$ and $p_{10}$ lie on consecutive rays around $p_0$, the points $p_9$ and $p_{10}$ must see each other, which is a contradiction.
  $\\ \\$
   Consider the other case where $p_0$ is the blocker of $(p_1,p_4)$. 
   Observe that a point $p_i$ on $\overrightarrow{p_0p_3}$ is required on $p_2p_4$ to block the visibility between
   $p_2$ and $p_4$.
   Similarly another point $p_j$ on $\overrightarrow{p_0p_2}$ is required on $p_1p_3$ to block the visibility between
   $p_1$ and $p_3$. 
   Moreover, $p_i$ and $p_j$ cannot be visible from $p_0$ unless they are $p_2$ and $p_3$.
   It can be seen that no pair of points $p_i$ and $p_j$ can satisfy these conditions, which is a contradiction.
\end{proof}
$\\$
The above lemmas show that Necessary Conditions \ref{nc1} and \ref{nc2} are not sufficient for recognizing a PVG,
which leads to Necessary Condition \ref{nc3}.
An assignment of vertex-blockers in $G$ is said to be a \emph{valid assignment} if it satisfies Necessary Conditions \ref{nc1} and \ref{nc2}.
Let $(v_i, v_{i,1}), (v_i, v_{i,2}), \ldots, (v_i, v_{i,d})$ be all visible pairs of $v_i$ in $G$.
For a valid assignment, let $S_{i,j}$ denote the set of vertices of $G$ such that for every vertex $u \in S_{i,j}$, $v_{i,j}$ is
a blocker assigned to the invisible pair $(v_i,u)$ in this assignment.

\begin{nc} \label{nc3}
 If $G$ is not a CSP, then there exists a valid assignment for $G$ such that 
for every vertex $v_i \in G$, there is an ordering of visible pairs $(v_i, v_{i,1}), (v_i, v_{i,1}), \ldots, (v_i, v_{i,d})$
around $v_i$ such that if $(v_i, v_{i,j})$ is adjacent to $(v_i, v_{i,k})$ in the ordering, then 
every vertex of $\{ v_{i,j} \} \cup S_{i,j}$ is adjacent to every vertex of $\{ v_{i,k} \} \cup S_{i,k}$ in $G$.
\end{nc} 
\begin{figure}[h] 
\begin{center}
\centerline{\hbox{\psfig{figure=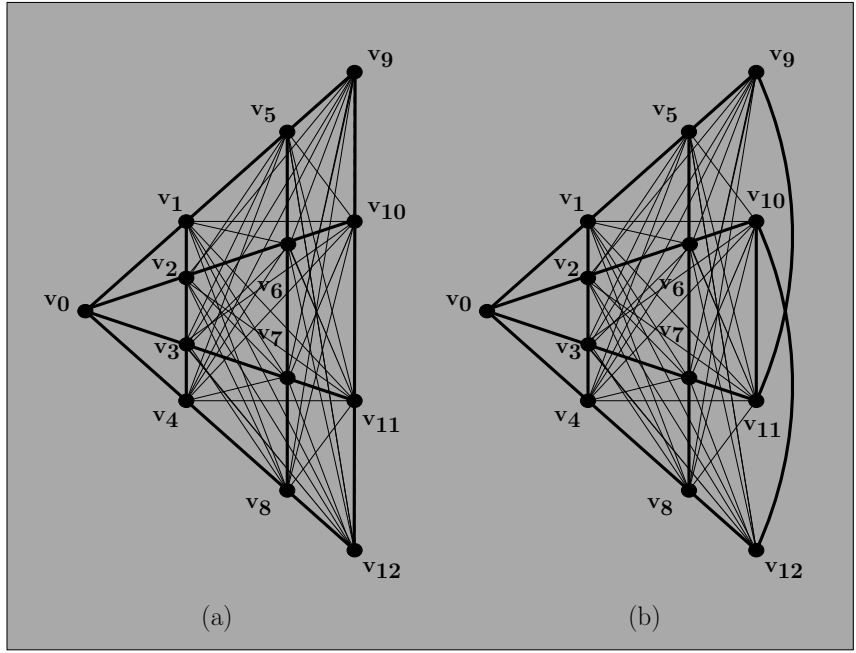,width=0.70\hsize}}}
\caption{ (a) This graph is a PVG drawn in the form of a visibility embedding. (b) This graph is not a PVG but satisfies
both Necessary Conditions \ref{nc1} and \ref{nc2}.}
 \label{figncprob}
\end{center}
\end{figure}
\begin{proof}
 Consider any valid assignment corresponding to a visibility embedding $\xi$ of $G$. 
 Let $\overrightarrow{p_ip_{i,j}}$ 
 denote the ray drawn from $p_i$ through $p_{i,j}$ in $\xi$. 
 Consider a clockwise ordering $A$ of $\overrightarrow{p_ip_{i,1}}, \ldots, \overrightarrow{p_ip_{i,d}}$ around $p_i$ in $\xi$
 such that the clockwise angle between any two rays in $A$ is convex, except possibly the last and first rays in $A$.
 So, every point on a ray in $A$ is visible from every point on its adjacent ray. 
 It can be seen that if any two rays $\overrightarrow{p_ip_{i,j}}$ and $\overrightarrow{p_ip_{i,k}}$ are adjacent in $A$, then every vertex
 of $\{ v_{i,j} \} \cup S_{i,j}$ is connected by an edge to every vertex of $\{ v_{i,k} \} \cup S_{i,k}$ in $G$.
 Hence, $G$ satisfies Necessary Condition \ref{nc3}.
%
%
\end{proof}
%
%
%
\begin{lemma} \label{mindeg}
If the size of the longest GSP in some visibility embedding of a graph $G$ with n vertices is k, then the degree 
of each vertex in $G$ is at least $\lceil \frac{n-1}{k-1} \rceil$ \cite{viscon-wood-2011, viscon-wood-2012, recogpvg-ghosh-roy-2011}.
\end{lemma}
\begin{proof}
For any point $p_i$ in a visibility
embedding of $G$, the degree of $p_i$ is the number of points visible from $p_i$ which are in angular 
order around $p_i$. Since the longest GSP is of size k, a ray from $p_i$ through any visible point of $p_i$ can contain
at most $k-1$ points excluding $p_i$. So there must be at least 
 $\lceil \frac{n-1}{k-1} \rceil$ such rays, which gives the degree of $p_i$.
  \end{proof}
\begin{theorem} \label{hamcyc}
 If $G$ is a PVG but not a path, then $G$ has a Hamiltonian cycle.
\end{theorem}
\begin{proof} Let $H_1,H_2,...,H_k$ be the convex layers of points in a visibility embedding of $G$, where
$H_1$ and $H_k$ are the outermost and innermost layers respectively. Let $p_ip_j$ be an edge of $H_1$,
where $p_j$ is the next clockwise point of $p_i$ on $H_1$ (Fig. \ref{ham1}(a)). Draw the left tangent of $p_i$ to $H_2$ meeting 
$H_2$ at a point $p_l$ such that the entire $H_2$ is to the left of the ray starting from $p_i$
through $p_l$. Similarly, draw the left tangent from $p_j$ to $H_2$ meeting $H_2$ at a point 
$p_m$. If $p_l=p_m$ then take the next clockwise point of $p_l$ in $H_2$ and call it $p_t$.
Remove the edges $p_ip_j$ and $p_lp_t$, and add the edges $p_ip_l$ and $p_jp_t$ (Fig.  \ref{ham1}(a)). 
Consider the other situation where $p_l \neq p_m$. If $p_lp_m$ is an edge, then remove the edges 
$p_ip_j$ and $p_lp_m$, and add the edges $p_ip_l$ and $p_jp_m$ (Fig.  \ref{ham1}(b)).  
If $p_lp_m$ is not an edge of $H_2$, take the next counterclockwise point of $p_m$ on $H_2$
and call it $p_q$. Remove the edges $p_ip_j$ and $p_qp_m$, and add the edges $p_ip_q$ and $p_jp_m$ (Fig. \ref{ham2}(a)).

\begin{figure} 
\begin{center}
\centerline{\hbox{\psfig{figure=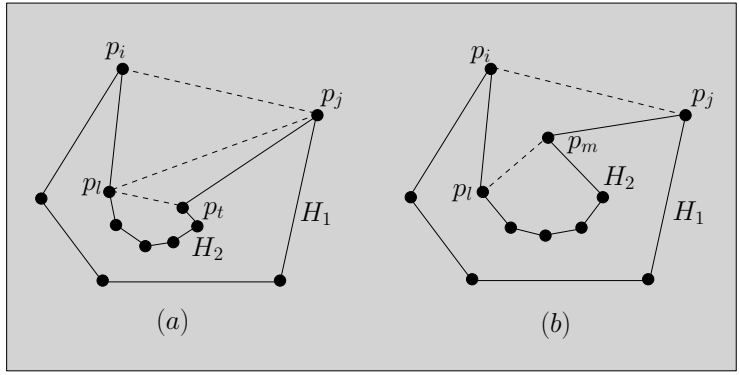,width=0.605\hsize}}}
\caption{ (a) The left tangents of $p_i$ and $p_j$ meet $H_2$ at the same point $p_l$. 
(b) The left tangents of $p_i$ and $p_j$ meet $H_2$ at points $p_l$ and $p_m$ of the same edge.}
\label{ham1}
\end{center}
\end{figure}

%
%
%
%
$\\ \\$
Thus, $H_1$ and $H_2$ are connected forming a cycle $C_{1,2}$. Without the loss of generality, we assume that 
$p_m \in H_2$ is the next counter-clockwise point of $p_j$ in $C_{1,2}$ (Fig. \ref{ham2}(b)). Starting from $p_m$, repeat the 
same construction to connect 
$C_{1,2}$ with $H_3$ forming $C_{1,3}$. Repeat till all layers are connected to form a Hamiltonian cycle $C_{1,k}$.
Note that if $H_k$ is just a path (Fig.  \ref{ham2}(b)), it can be connected trivially to form $C_{1,k}$.   \end{proof}
\begin{figure}  
\begin{center}
\centerline{\hbox{\psfig{figure=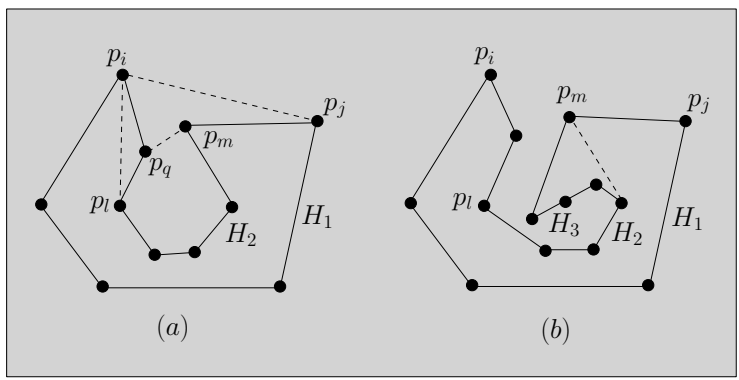,width=0.605\hsize}}}
\caption{ (a) The left tangents of $p_i$ and $p_j$ meet $H_2$ points $p_l$ and $p_m$ of different edges. 
 (b) The innermost convex layer is a path which is connected to $C_{1,2}$.}
\label{ham2}
\end{center}
\end{figure}
\begin{corollary} \label{hampoly}
 Given $G$ and a visibility embedding of $G$, a Hamiltonian cycle in $G$ can be constructed in linear time. 
\end{corollary}
\begin{proof} This is because the combinatorial representation of G contains all its edges, and hence the gift-wrapping 
algorithm for finding the convex layers of a point set becomes linear in the input size.   \end{proof}
\begin{lemma} \label{ineq} Consider a visibility embedding of $G$. Let A, B and C be three nonempty, disjoint sets of points in it
 such that $\forall p_i \in A $ and $\forall p_j \in C$, the GSP between $p_i$ and $p_j$ contains at least one point from B,
and no other point from A or C (Fig. \ref{ineqcons}(a)).
 Then $ |B| \geq |A| + |C| - 1$ \cite{viscon-wood-2011, viscon-wood-2012, recogpvg-ghosh-roy-2011}.
\end{lemma}
 
\begin{figure} 
\begin{center}
\centerline{\hbox{\psfig{figure=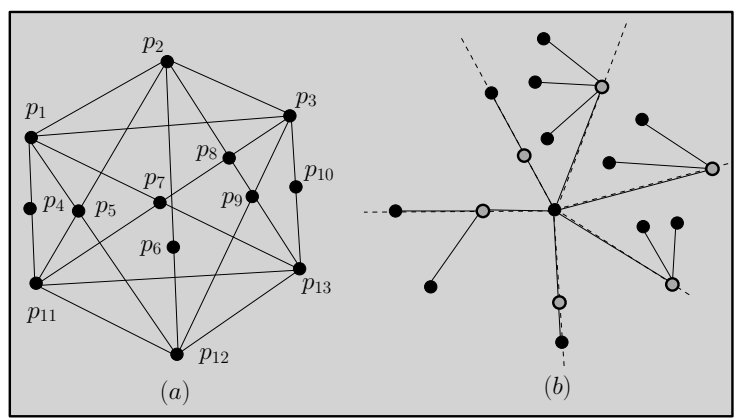,width=0.60\hsize}}}
\caption{  (a) A PVG with A = $\{p_1$, $p_2$, $p_3\}$, 
B=$\{p_4$, $p_5$, $p_6$, $p_7$, $p_8$, $p_9$, $p_{10}\}$ and 
C=$\{p_{11}$, $p_{12}$, p$_{13}\}$. (b) Points of A and C connected by edges representing blockers. }
\label{ineqcons}
\end{center}
\end{figure}

\begin{proof} Draw rays from a point $p_i \in A$ through every point of C (Fig. \ref{ineqcons}(b)). 
These rays partition the plane into $|C|$ wedges. Since points of C are not visible from $p_i$, there is at least 
one blocker lying on each ray   between $p_i$ and the point of C on the ray. So, there are at least
$|C|$ number of such blockers. Consider the remaining $|A-1|$ points of A lying in different wedges.
Consider a wedge bounded  by two rays drawn through $p_k\in C$ and $p_l\in C$.
Consider the segments from $p_k$ to all points of A in the wedge. Since these segments meet only at 
$p_k$, and $p_k$ is not visible from any point of A in the wedge, each of these segments must contain a distinct blocker.
So, there are at least $|A|-1$ blockers in all the wedges. Therefore the total number of points
in B is at least  $|A| + |C| - 1$.   \end{proof}  
\begin{lemma}
 
\label{genineq} Consider a visibility embedding of $G$. Let A and C be two nonempty and disjoint sets of points  
 such that 
no point of A is visible from any point of C.
  Let B be the set of points (or blockers) on the segment  $p_i p_j$, $\forall p_i \in A $ 
and $\forall p_j \in C$, and blockers in B are allowed to be points of A or C. Then $ |B| \geq |A| + |C| - 1$ \cite{recogpvg-ghosh-roy-2011}.
 
\end{lemma}
\begin{proof}  Draw rays from a point $p_i \in A$ through every point of C. 
These rays partition the plane into at most $|C|$ wedges.
Consider a wedge bounded  by two rays drawn through $p_k\in C$ and $p_l\in C$.
Since these rays may contain other points of A and C, all points  
 between $p_i$ and the farthest point from $p_i$ on a ray, are blockers in B. 
Observe that all these blockers except one may be from A or C.
 Thus,  excluding $p_i$, B has at least 
as many points as from A and C on the ray.
Consider the points of A inside the wedge.
Draw segments from $p_k$ to all points of A in the wedge.
Since these segments may contain multiple points from A, all points on 
a segment between $p_k$ and the farthest point from $p_k$ are blockers in B. All these points except one may be from A.
Thus, B has at least 
as many points as from A inside the  wedge. Therefore the total number of points
in B is at least  $|A| + |C| - 1$.    \end{proof}
\section{Computational complexity of the recognition problem}
 In this section we show that the recognition problem for a PVG lies in PSPACE.
Our technique in the proof follows a similar technique used by Everett \cite{ev-vgc-90} for
showing that the recognition problem for polygonal visibility is in PSPACE.
We start with the following theorem of Canny \cite{can-pspace}.  
\begin{theorem}\label{cannypspace}
Any sentence in the existential theory of the reals can be decided in PSPACE.\end{theorem}
$\\ $
A sentence in the first order theory of the reals is a formula of the form :
$$ \exists x_1 \exists x_2 ... \exists x_n \textit{P} (x_1 , x_2 , ... , x_n ) $$  
where the $x_i's$ are variables ranging over the real numbers and where   $\textit{P} (x_1 , x_2 , ... , x_n ) $ is a 
predicate built up from $\neg$, $\wedge$,  $\vee$, =, $<$, $>$ , +, $\times$, 0, 1 and -1 in the usual way.

\begin{theorem} \label{recogPSPACE}
The recognition problem for point visibility graphs lies in PSPACE.\end{theorem}
\begin{proof} Given a graph $G(V,E)$, we construct a formula in the existential theory of the reals polynomial in size of $G$
which is true if and only if $G$ is a point visibility graph. 
$\\ \\$
Suppose $(v_i, v_j) \notin E $. This means that if $G$ admits a visibility embedding,
 then there must be a blocker (say, $p_k$) on the segment joining 
  $p_i$ and $p_j$.
  Let the 
coordinates of the points   $p_i$, $p_j$ and $p_k$ be $(x_i, y_i)$, $(x_j, y_j)$ and $(x_k, y_k)$
respectively. So we have :
$\\ \\$
$\exists t_ \in \mathbbm{R} \Big{(} \big{(}0 < t\big{)} \wedge \big{(}t
 < 1\big{)} \wedge \big{(}(x_k-x_i) = t \times 
(x_j-x_i)\big{)} \wedge \big{(} (y_k-y_i) = t \times (y_j-y_i)\big{)}\Big{)} $
$\\ \\$
Now suppose $(v_i, v_j) \in E $. This means that if $G$ admits a visibility embedding, 
 no point in $P$ lies on the segment connecting 
$p_i$ and $p_j$ to ensure visibility. So, (i) either $p_k$ forms
a triangle with $p_i$ and $p_j$ or (ii) $p_k$ lies on the line passing
through $p_i$ and $p_j$ but not  between $p_i$ and $p_j$.
Determinants of non-collinear points is non-zero. So we have :
$\\ \\$
$\exists t \in \mathbbm{R} \Big{(} \big{(}det(x_i,x_j,x_k,y_i,y_j,y_k) > 0 \big{)}\vee 
\big{(} det(x_i,x_j,x_k,y_i,y_j,y_k) < 0 \big{)} \Big{)} \bigvee  
 \Big{(} \big{(} t > 1 \big{)}\vee\big{(}  t < -1  \big{)}\wedge  \big{(} 
(x_k-x_i) = t \times  (x_j-x_i)\big{)} \wedge \big{(} (y_k-y_i) 
= t \times (y_j-y_i)\big{)}\Big{)} $ 
$\\ \\$
For each triple $(v_i,v_j,v_k)$ of vertices in $V$, we add a $t=t_{i,j,k}$ to the existential part of the formula and 
the corresponding portion to the predicate. So the formula becomes:
$\\ \\$
$ \exists x_1 \exists y_1 ... \exists x_n \exists y_n \exists t_{1,2,3} .... \exists t_{n-2,n-1,n}$ 
$\textit{P}(x_1,y_1,...,x_n,y_n, t_{1,2,3},...,t_{n-2,n-1,n})$
$\\ \\$
which is of size $O(n^3)$. This proves our theorem.   \end{proof}
 \begin{figure} 
\begin{center}
\mbox{\begin{minipage} [b] {70mm}
\centerline{\hbox{\psfig{figure=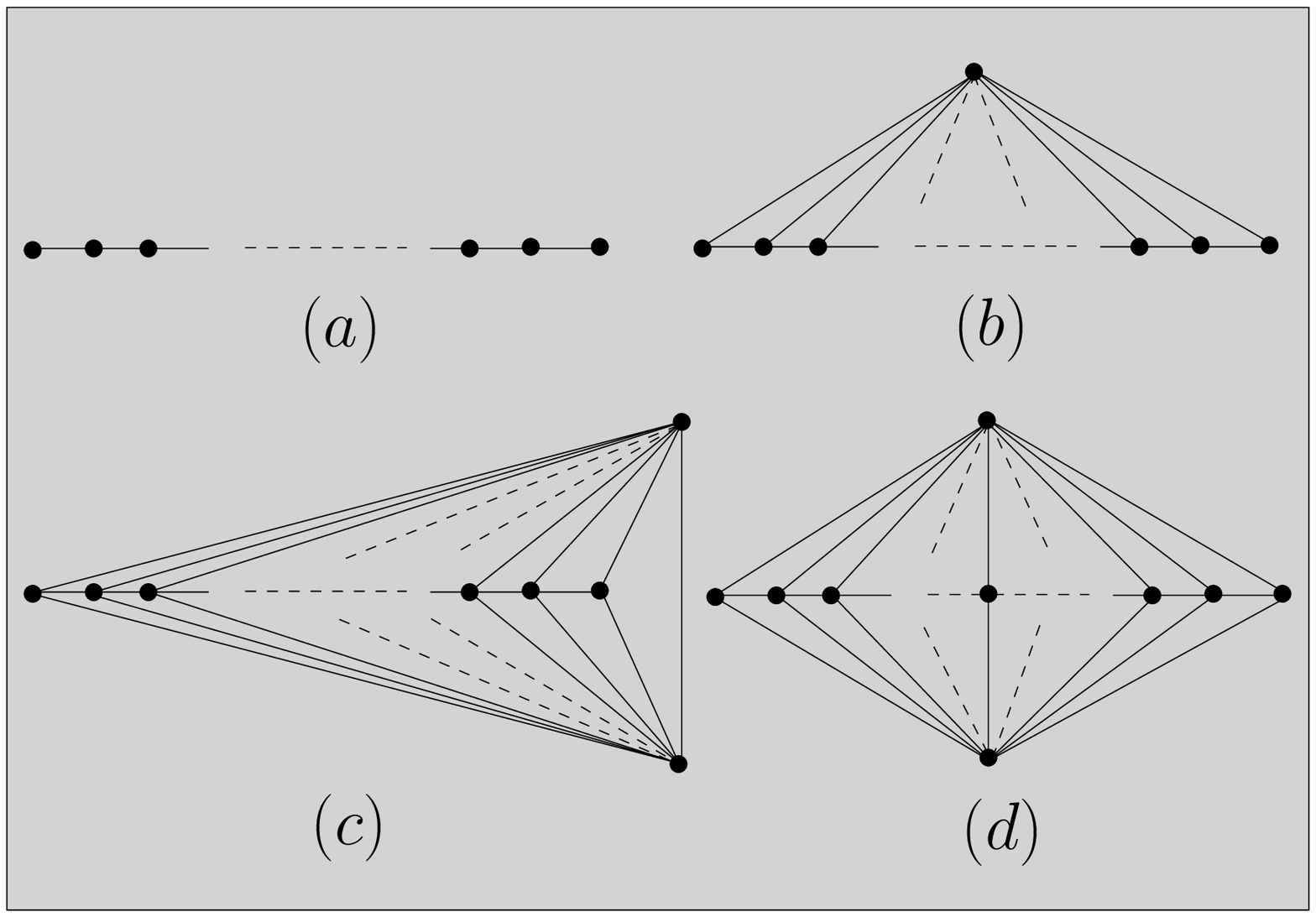,width=0.955\hsize}}}
\caption{ These four infinite families admit planar visibility embedding (Eppstein \cite{epp-plpvg}).}
\label{4inf}
 \end{minipage}}\hspace{2mm}
 \mbox{\begin{minipage} [b] {48mm}
\centerline{\hbox{\psfig{figure=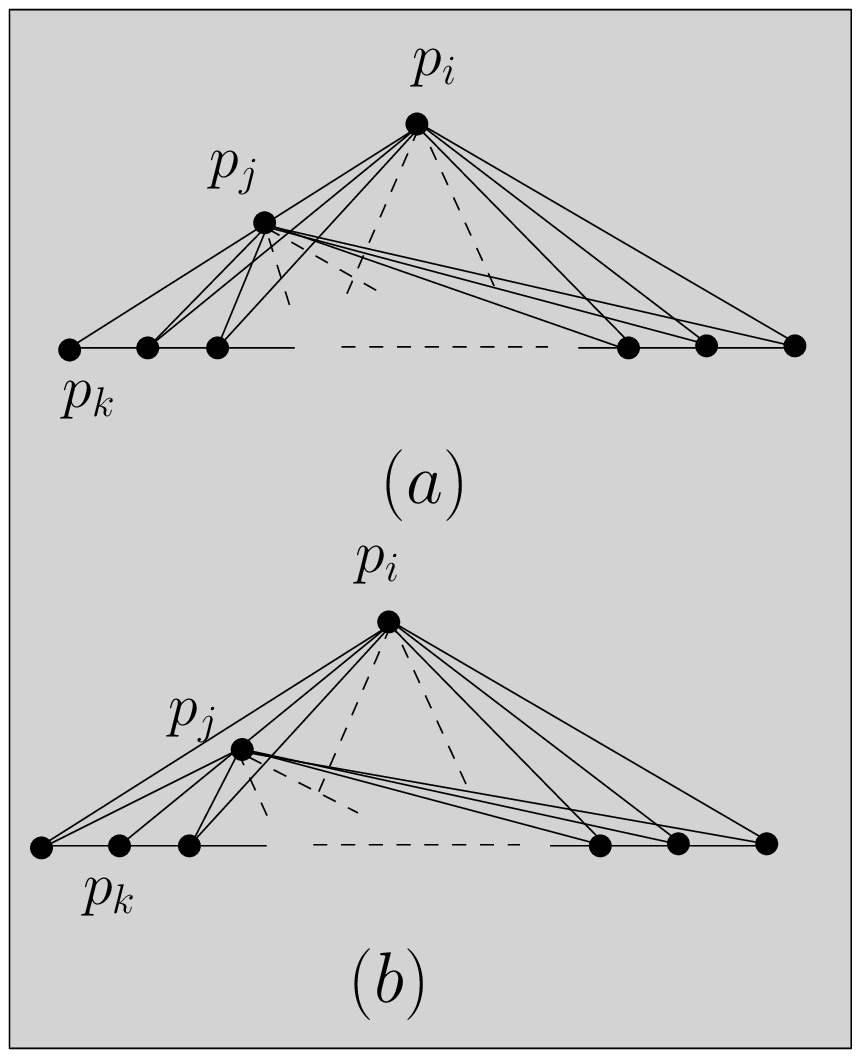,width=0.85\hsize}}}
\caption{ These two infinite families do not admit planar visibility embedding.}
\label{2inf}
\end{minipage}}
\end{center}
\end{figure}
\section{Planar point visibility graphs }
In this section, we present  a characterization, recognition and reconstruction of planar point visibility graphs.
Let $G$ be a given planar graph. 
We know that the planarity of $G$ can be tested
in linear time \cite{lin-plan-2004}.
If $G$ is planar, a straight line embedding of $G$ can also be constructed in
linear time. However, this embedding may not satisfy the required visibility constraints, and therefore, it cannot be a 
visibility embedding.
We know that collinear points play a crucial role in a visibility embedding of $G$.
It is, therefore, important to identify points belonging to a GSP of maximum length. Using this approach, 
  we construct a visibility embedding of a given planar graph $G$, if it exists.
We have the following lemmas on visibility embeddings of $G$.
\begin{lemma} \label{genbnd}
 Assume that $G$ admits a visibility embedding $\xi$. If $\xi$ has at least one $k$-GSP for $k \geq 4 $, then
the number of vertices in $G$ is at most $$ k + \Big{\lfloor} \frac{2k-5}{k-3} \Big{\rfloor}$$  
\end{lemma}
\begin{proof} By Lemma \ref{deg}, $G$ can have at least $(k-1) + (n-k)k$ edges. By applying Euler's criterion for planar graphs,
we have the following inequality on the 
number of permissible edges of $G$. 
\begin{eqnarray}
 (k-1) + (n-k)k &\leq& 3(n) - 6 \nonumber \\
\Rightarrow (k-1) + (n-k)k &\leq& 3(k+n-k) - 6 \nonumber \\
  \Rightarrow  (k-1)+(n-k)k   &\leq& 3k + 3(n-k) - 6 \nonumber \\
  \Rightarrow  (n-k)(k-3)\ \ \ \ \  &\leq& 2k-5 \nonumber  \\
  \Rightarrow (n-k)\ \ \ \ \ \ \ \ \ \ \ \ \ \  &\leq& \frac{2k-5}{k-3} \nonumber \\ 
\end{eqnarray} 
 Since $(n-k)$ must be an integer, we have
\begin{eqnarray} 
 (n-k)      &\leq& \Big{\lfloor} \frac{2k-5}{k-3} \Big{\rfloor} \nonumber \\ 
  \Rightarrow n \ \ \ \  \ \ \ &\leq&  k + \Big{\lfloor} \frac{2k-5}{k-3} \Big{\rfloor}
 \end{eqnarray}     \end{proof}


\begin{figure} 
\begin{center}
\centerline{\hbox{\psfig{figure=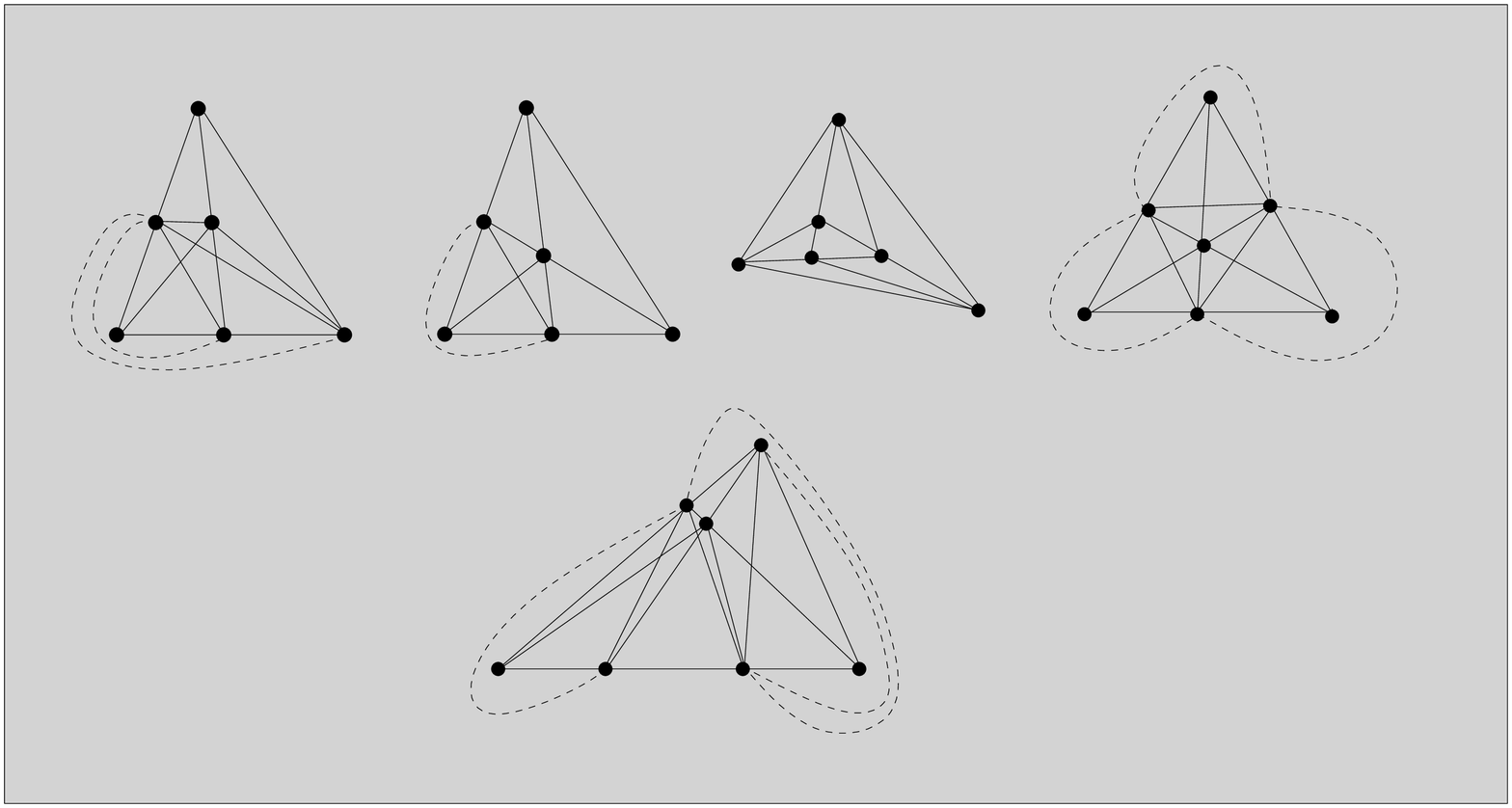,width=0.70\hsize}}}
\caption{ Five planar PVGs that do not belong to any of the six families.
Dotted lines show how the edge-crossings in the visibility embedding can be avoided in a planar embedding.}
\label{9ppvg}
\end{center}
\end{figure}
\begin{corollary} \label{sixfam}
  There are six infinite families of planar graphs $G$ that admit a visibility embedding $\xi$ with a
 $k$-GSP for $k \geq 5$ (Figs. \ref{4inf} and \ref{2inf}).
 \end{corollary}
\begin{proof} For $k\geq 5$, $n \leq k + 2$. There can be only six
infinite families of graphs having at most two points outside a maximum size GSP in $\xi$ (denoted as $l$)
as follows. 
\begin{enumerate}
 \item There is no point lying outside $l$ in $\xi$ (see Fig. \ref{4inf}(a)).
\item There is only one point lying outside $l$ in $\xi$ that is adjacent to all points in $l$ (see Fig. \ref{4inf}(b)).
\item There are two points lying outside $l$ in $\xi$ that are adjacent to all other points in $\xi$ (see Fig. \ref{4inf}(c)).
\item There are two points lying outside $l$ in $\xi$ 
that are not adjacent to each
other but adjacent to all points of $l$ in $\xi$ (see Fig. \ref{4inf}(d)).
\item There are two points $p_i$ and $p_j$ lying outside $l$ in $\xi$ 
such that $p_i$ and $p_j$ are adjacent to all other points in $\xi$ except an endpoint $p_k$ of $l$  
 as $p_j$ is a blocker on $p_ip_k$ (see Fig. \ref{2inf}(a)).
\item Same as the previous case, except $p_k$ is now an intermediate point of $l$ in $\xi$ (see Fig. \ref{2inf}(b)).
\end{enumerate}   \end{proof} 
Let us identify those graphs that do not belong to these six infinite families. We show in the following that such graphs
can have a maximum of eight vertices.
\begin{lemma} \label{4csp}
Assume that $G$ admits a visibility embedding $\xi$. If $\xi$ has at least one 4-GSP, then
the number of vertices in $G$ is at most seven.
\end{lemma}
\begin{proof} Putting $k=4$ in the formula of Lemma \ref{genbnd}, we get $n \leq 7$.   \end{proof}
\begin{lemma} \label{3csp}
 Assume that $G$ admits a visibility embedding $\xi$. If $G$ has at least one 3-CSP but no 4-CSP,
then $G$ has at most eight vertices.
\end{lemma}
\begin{proof} Since $G$ has no 4-CSP, and $G$ is not a clique, there is a 3-GSP in $\xi$.
Starting from the 3-GSP, points are added one at a time to construct $\xi$.
Since no subsequent point can be added on the line passing through points of the 3-GSP to prevent forming a 4-GSP,
adding the fourth and fifth points gives at least three edges each in $\xi$.
As $\xi$ does not have a 4-CSP, there can be at most one blocker between an invisible pair of points in $\xi$.
So, for the subsequent points, at least $\lceil \frac{i-1}{2} \rceil$ edges
are added for the $ith$ point.
Since $G$ is planar, by Euler's condition we must have:
$  8 + \displaystyle\sum^n_{i= 6}\Big\lceil \frac{i-1}{2} \Big\rceil \leq 3n-6$. This inequality is valid only up to $n=8$.  
\end{proof}


\begin{lemma} \label{9planar}
 There are five distinct planar graphs $G$ that admit visibility embeddings but do not belong to the six 
 infinite families (Fig. \ref{9ppvg}).
\end{lemma}
\begin{theorem} \label{totalchar}
 Planar point visibility graphs can be characterized by six infinite families of graphs and five particular graphs.
\end{theorem}
\begin{proof} Five particular graphs can be identified by enumerating all points of eight vertices as shown 
 in Fig. \ref{9ppvg}. For the details of the enumeration, see the appendix.
  \end{proof}
\begin{theorem} \label{polyrecog}
 Planar point visibility graphs can be recognized in linear time.
\end{theorem}
\begin{proof} Following Theorem \ref{totalchar}, $G$ is tested initially whether it is isomorphic to any of the six 
particular graphs for $n \leq 8$. Then, the maximum CSP is identified before its adjacency is tested
with the remaining one or two vertices of $G$. 
The entire testing can be carried out in linear time.   \end{proof}
\begin{corollary}
 Planar point visibility graphs can be reconstructed in linear time.
\end{corollary}
\begin{proof} Theorem \ref{polyrecog} gives the relative positions and collinearity of points in the visibility embedding of $G$.
Since each point can be drawn with integer coordinates of size $O(logn)$ bits, $G$ can be reconstructed in linear time.  \end{proof}
\section{Concluding remarks}
We have presented three necessary conditions for recognizing point visibility graphs.
 Though the first necessary condition can be tested in $O(n^3)$ time, it is not clear how vertex-blockers
can be assigned to every invisible pair in $G$ in polynomial time 
satisfying the second necessary condition. Observe that these assignments in a visibility embedding
give the ordering of collinear points along any ray  starting from any point through its 
visible points. These rays together form an arrangement of rays in the plane. 
It is open whether such an arrangement can be constructed satisfying
assigned vertex-blockers in polynomial time.
The third necessary condition gives the ordering of these rays around each point. It is also not clear
whether the third necessary condition can be tested in polynomial time. Overall, we feel that the three
necessary conditions may be sufficient.
$\\ \\$
Let us consider the complexity issues of the problems of Vertex Cover, Independent Set and Maximum Clique in a point visibility graph.
Let $G$ be a graph of $n$ vertices, not necessarily a PVG.  We construct another graph $G'$ such that 
(i) $G$ is an induced subgraph of $G'$, and (ii) $G'$ is a PVG.
Let $C$ be a convex polygon drawn along with all  its diagonals,
where every vertex $v_i$ of $G$ corresponds to a vertex $p_i$ of $C$. 
For every edge $(v_i,v_j)  \notin  G$, introduce a blocker $p_t$ on
the edge $(p_i,p_j)$ such that $p_t$ is visible to all points of $C$ and 
all blockers added so far. 
Add edges from $p_t$ to all vertices of $C$ and blockers in $C$.
The graph corresponding to this embedding is called $G'$.
So, $G'$ and its embedding can be constructed in polynomial time. 
Let the sizes of the minimum vertex cover, maximum independent set and maximum clique in 
$G$ be $k_1$, $k_2$ and $k_3$ respectively. If $x$ is the number of blockers added to $C$, then the sizes of the 
minimum vertex cover, maximum independent set and maximum clique in 
$G'$ are $k_1 + x $, $k_2$ and $k_3 + x$ respectively. Hence, the problems remain NP-Hard.
\begin{theorem}
The problems of Vertex Cover, Independent Set and Maximum Clique remain NP-hard on point visibility graphs.
\end{theorem}
$\\ \\$
\textbf{Acknowledgements}  
$\\ \\$
The preliminary version  
of a part of this work was submitted 
in May, 2011 as a Graduate School Project Report of Tata Institute of Fundamental Research 
\cite{recogpvg-ghosh-roy-2011}.
The authors would like to thank Sudebkumar Prasant Pal for his helpful comments during the
preparation of the first version of the manuscript \cite{recogpvg-2012}.


\bibliographystyle{plain}
\bibliography{vis}

\section*{Appendix}
 
By enumeration,  we identify all five particular graphs (see Fig. \ref{9ppvg}) that do not belong to the six infinite families
(see Figs. \ref{4inf} and \ref{2inf}),
as stated in Theorem \ref{totalchar}. We know from Lemmas \ref{4csp} and \ref{3csp} that $n \leq 8$.
We have the following cases.
 
\medskip

\noindent \textbf{Case 1.}
{\it There is a 3-GSP but no 4-GSP in some visibility embedding $\xi$ of $G$.}


\medskip

\noindent If $n \leq 5$, $G$ belongs to one of the infinite families having at most two points outside the 3-GSP.
%
%
%
%
%
%
%
$\\ \\$
Consider $ n = 6$. Let $p_1$, $p_2$ and $p_3$ be collinear points representing a 3-GSP (denoted as $l$). If there is 
no other 3-GSP in $\xi$, then all edges except $(v_1,v_3)$ are present in $G$. So, $G$ is not planar as it has $K_5$
as a subgraph. If there is another 3-GSP (say, $l'$) in $\xi$, which is disjoint from $l$, then $G$ is not planar
as it has $K_{3,3}$ as a subgraph. So, we consider the situation when $l$ and $l'$ share a point in $\xi$. There can be three such 
distinct embeddings of five points as shown in Fig. \ref{3gspopp}. Before the sixth point $p_6$ is added in the embeddings, we
need the following lemma.
$\\ \\$
\begin{figure} 
\begin{center}
\centerline{\hbox{\psfig{figure=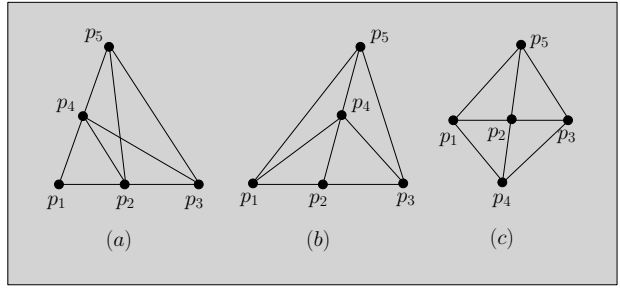,width=0.505\hsize}}}
\caption{Visibility embeddings of five points containing two overlapping 3-GSPs. }
\label{3gspopp}
\end{center}
\end{figure}
%

\begin{lemma}\label{app1}
 Any planar point visibility graph $H$ of six vertices, with no 4-GSP, has at least three 3-CSPs.
\end{lemma}
\begin{proof} We know that if $H$ does not have an edge between two vertices, then it corresponds to a 3-CSP.
Since $H$ has at most $12$ edges due to Euler's condition, and a complete graph on six vertices 
has $15$ edges, there are at least 3 edges not present in $H$. Therefore $H$ has at least three 3-CSPs.   \end{proof}

\begin{figure}   
\begin{center}
\centerline{\hbox{\psfig{figure=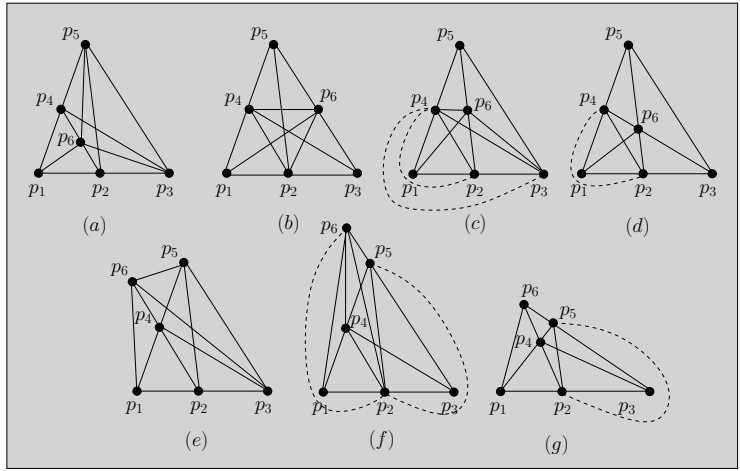,width=0.605\hsize}}}
\caption{Visibility embeddings of six points after $p_6$ is added to the embedding in Fig. \ref{3gspopp}(a).
Dotted lines show how the edge-crossings in the visibility embedding can be avoided in a planar embedding.}
\label{3gspsm}
\end{center}
\end{figure}
%


$\\ \\$
Let us add $p_6$ to the embedding shown in Fig. \ref{3gspopp}(a) in such a way that the new embeddings have three 3-GSPs
satisfying Lemma \ref{app1}. So, $p_6$ must lie on the lines passing through exactly two points, forming a new
3-GSP. Removing symmetric embeddings, we have the 
following choices of positioning $p_6$ in the new 3-GSP: $\overline {p_4 p_6 p_2}$ (Fig. \ref{3gspsm}(a)),
 $\overline {p_5 p_6 p_3}$ (Fig. \ref{3gspsm}(b)), $\overline {p_5 p_6 p_2}$ (Fig. \ref{3gspsm}(c)), 
$\overline {p_5 p_6 p_2}$ and $\overline {p_4 p_6 p_3}$ (Fig. \ref{3gspsm}(d)), 
$\overline {p_6 p_4 p_2}$ (Fig. \ref{3gspsm}(e)),
 $\overline {p_6 p_5 p_3}$ (Fig. \ref{3gspsm}(f)), $\overline {p_6 p_4 p_2}$  and $\overline {p_6 p_5 p_3}$ (Fig. \ref{3gspsm}(g)).
It can be seen that embeddings in Figs.  \ref{3gspsm}(a), \ref{3gspsm}(b) and \ref{3gspsm}(e) correspond to non-planar graphs, and 
embeddings in Figs.  \ref{3gspsm}(c),  \ref{3gspsm}(d),  
\ref{3gspsm}(f) and  \ref{3gspsm}(g) correspond to planar graphs. Graphs corresponding to embeddings in Figs. \ref{3gspsm}(c)
  and \ref{3gspsm}(d), are isomorphic to graphs corresponding to embeddings in Figs. \ref{3gspsm}(f) 
and  \ref{3gspsm}(g) respectively. Hence, only two non-isomorphic planar
graphs arise after adding $p_6$ to the visibility embedding in Fig. \ref{3gspopp}(a).
%
%
%
$\\ \\$
As before, let us add $p_6$ to the embedding shown in Fig. \ref{3gspopp}(b). Removing symmetric embeddings, we have the 
following choices of positioning $p_6$ in the new 3-GSP: $\overline {p_1 p_6 p_5}$ (Fig. \ref{new11}(a)),
$\overline {p_1 p_5 p_6}$ (Fig. \ref{new11}(b)), $\overline {p_6 p_1 p_5}$ (Fig. \ref{new11}(c)),
$\overline {p_1 p_6 p_5}$ and $\overline {p_3 p_4 p_6}$ (Fig. \ref{new11}(d)),
$\overline {p_1 p_6 p_4}$ (Fig. \ref{new11}(e)),
$\overline {p_6 p_1 p_4}$ (Fig. \ref{new11}(f)) and $\overline {p_1 p_4 p_6}$ (Fig. \ref{new11}(g))
The embeddings in all the figures except Figure \ref{new11}(f) have two 3-GSPs that overlap at their end-points,
which they are already considered in Fig. \ref{3gspsm}. Since the embedding in Fig. \ref{new11}(f) is planar,
this is the only new planar
graph that arises after adding $p_6$ to the visibility embedding in Fig. \ref{3gspopp}(b).
\begin{figure}   
\begin{center}
\centerline{\hbox{\psfig{figure=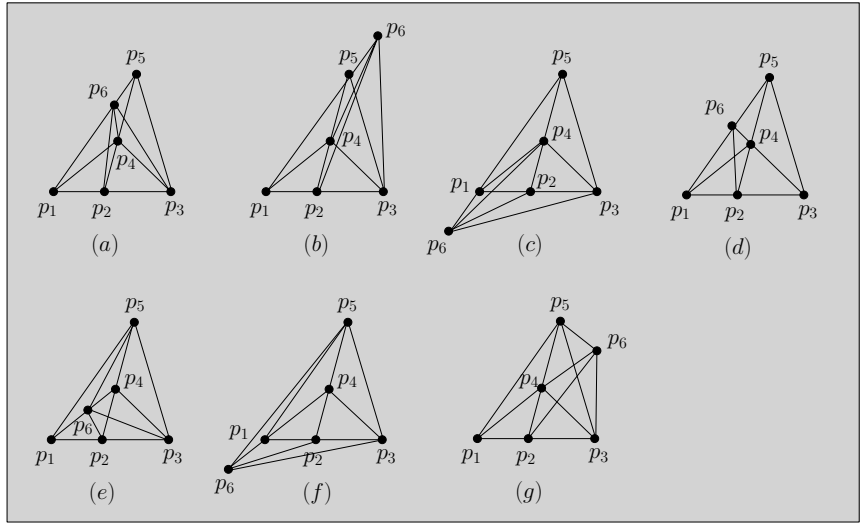,width=0.705\hsize}}}
\caption{Visibility embeddings of six points after $p_6$ is added to the embedding in Fig. \ref{3gspopp}(b).}
\label{new11}
\end{center}
\end{figure}
%
%
$\\ \\$
As before, let us add $p_6$ to the embedding shown in Fig. \ref{3gspopp}(c). Removing symmetric embeddings, we have the 
following choices of positioning $p_6$ in the new 3-GSP: $\overline {p_1 p_6 p_5}$ (Fig. \ref{new12}(a)) and
$\overline {p_1 p_5 p_6}$ (Fig. \ref{new12}(b)). But these two embeddings are already present in Fig. \ref{3gspsm}.
So, no new planar graphs arise after adding $p_6$ to the embedding visibility in Fig. \ref{3gspopp}(c). 
Thus,
three particular planar point-visibility graphs of six vertices are identified (see Figs. \ref{3gspsm}(c), 
\ref{3gspsm}(d) and \ref{new11}(f)).
 Consider $n = 7$. 
 In the following lemma, we show that there is exactly one particular graph of seven vertices that admits a planar embedding
(Fig. \ref{new16}).



\begin{lemma}
Let $H$ be a planar point visibility graph on seven vertices such that it has a 3-GSP but no 4-GSP in every
visibility embedding $\xi$ of $H$. Then $\xi$ has exactly six 3-GSPs.
\end{lemma}
\begin{proof} 
Since $H$ has at most $15$ edges due to Euler's condition, and a complete graph on seven vertices 
has $21$ edges, there are at least six invisible pairs in $H$.
So, $H$ has at least six 3-GSPs in $\xi$,
On the other hand, if $\xi$ has seven 3-GSPs, then
there are seven invisible pairs in $H$. So, $H$ can have maximum of 14 edges. But then, every line in $\xi$
must pass through exactly three points, contradicting Sylvester-Gallai Theorem \cite{syl-sur-90}.  
  \end{proof}
\begin{corollary}
 If $p_7$ is added to the embeddings of particular graphs of six vertices in Figs. \ref{3gspsm}(c),
\ref{3gspsm}(d) and \ref{new11}(f), then only one embedding gives rise to a  planar embedding as
shown in Fig. \ref{new16}.
\end{corollary}
 \noindent Consider $ n = 8$. In the following lemma, we show that there is no particular graph on eight vertices.
\begin{lemma}
 There is no particular planar point visibility graph on eight vertices that has a 3-CSP but no 4-CSP.
\end{lemma}
 \begin{proof}
   We know that if $G$ does not have an edge between two vertices, then it corresponds to a 3-CSP.
Since $G$ has at most $18$ edges due to Euler's condition, and a complete graph on eight vertices 
has $28$ edges, there are at least ten edges not present in $G$. Therefore $G$ must have at least ten edge disjoint 3-CSPs.
But ten edge disjoint 3-CSPs require $20$ edges. Since $G$ can have at most $18$ edges, such a $G$ cannot exist.
 
 \end{proof}

\begin{figure}   
\begin{center}
\mbox{\begin{minipage} [b] {75mm}
\centerline{\hbox{\psfig{figure=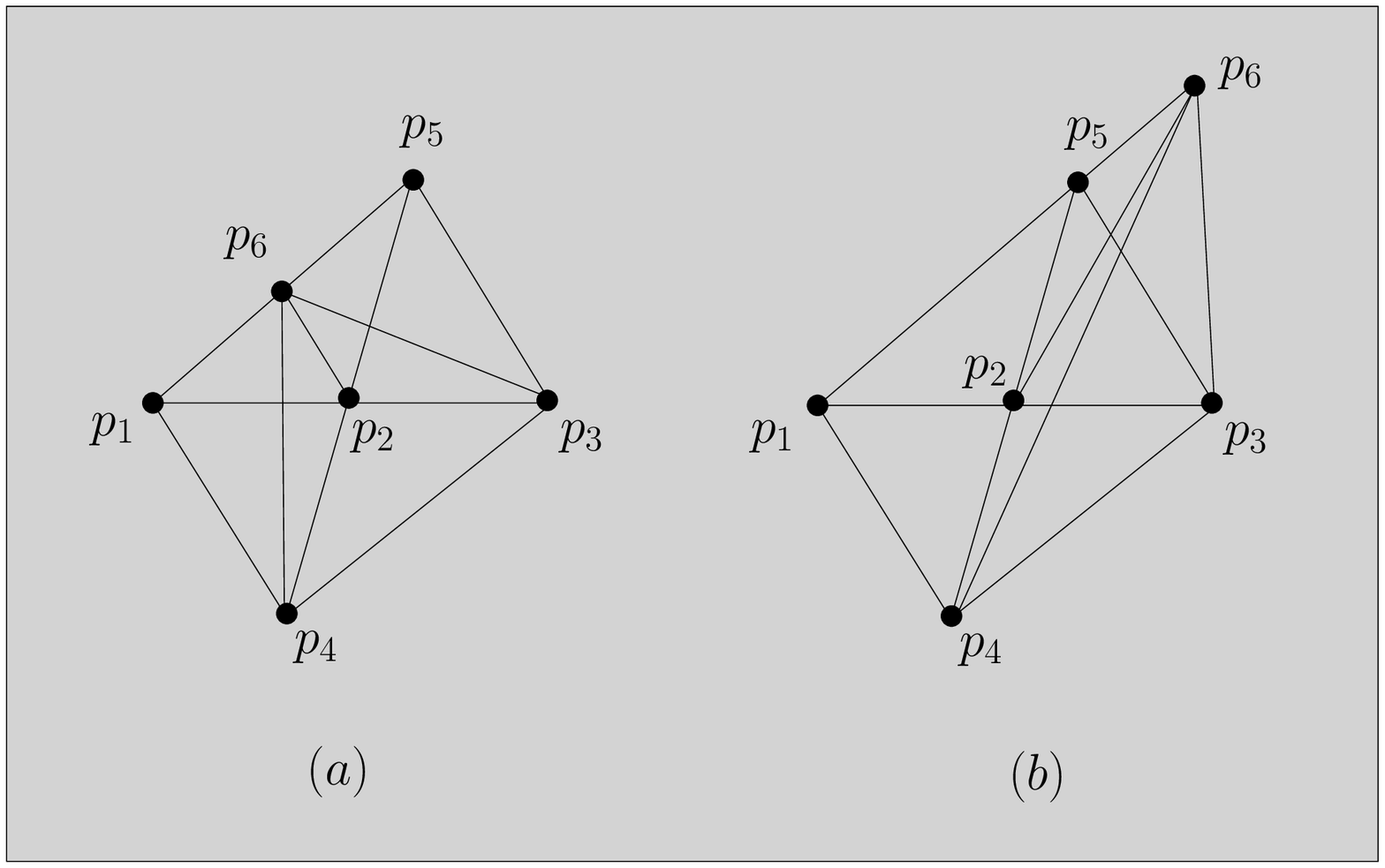,width=0.955\hsize}}}
\caption{Visibility embeddings of six points after $p_6$ is added to the embedding in Fig. \ref{3gspopp}(c).}
 \label{new12}
\end{minipage}}\hspace{2mm}
\mbox{\begin{minipage} [b] {75mm}
\centerline{\hbox{\psfig{figure=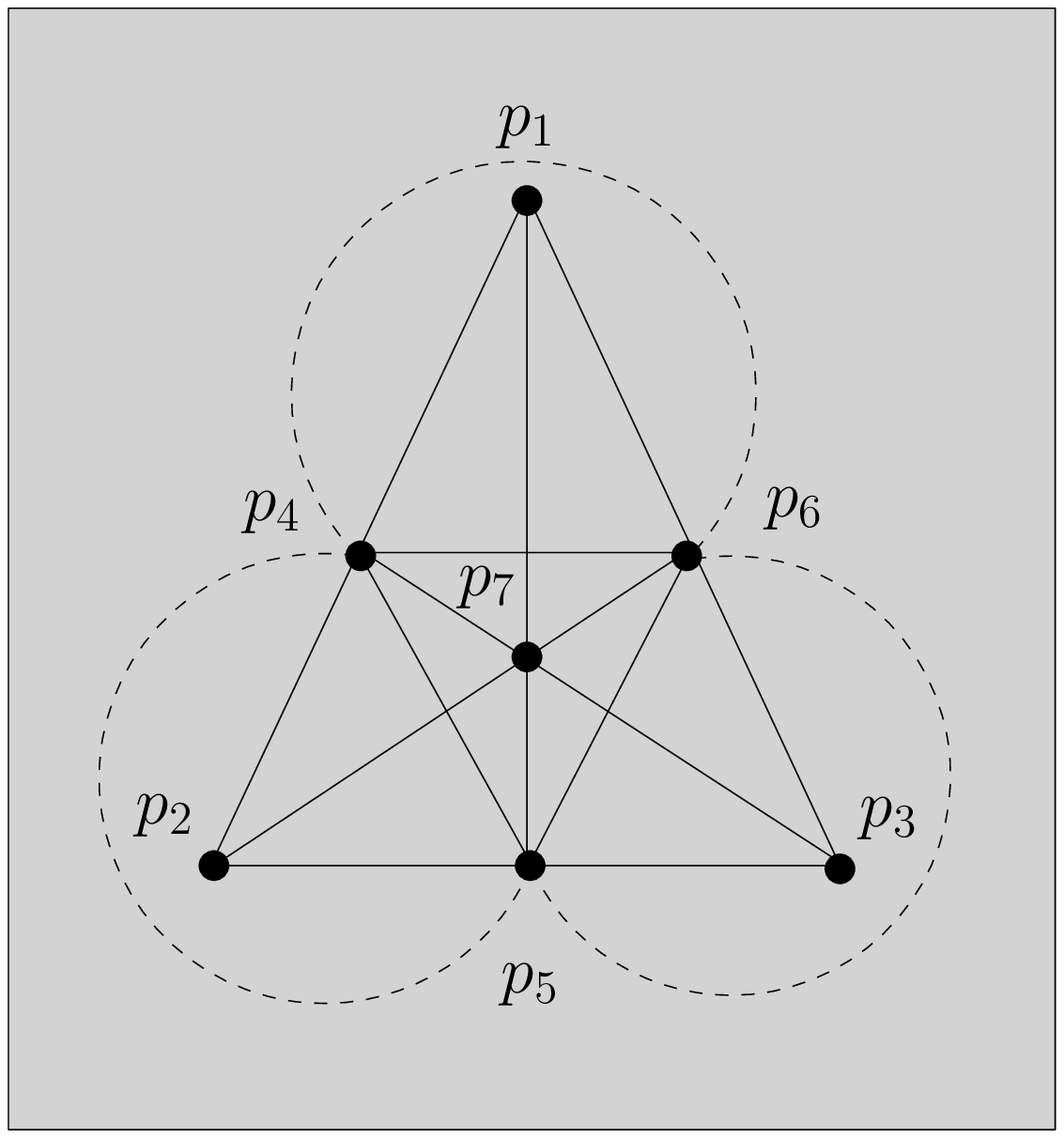,width=0.555\hsize}}}
\caption{Unique visibility embedding of planar point visibility graph on seven vertices, with a 3-GSP but no 4-GSP.
Dotted lines show how the edge-crossings in the visibility embedding can be avoided in a planar embedding.}
 \label{new16}
\end{minipage}}
\end{center}
\end{figure}
%

\noindent \textbf{Case 2.} {\it There is a 4-GSP but no 5-GSP in every visibility embedding of $G$.}
$\\ \\$
If $n \leq 6$, $G$ belongs to one of the infinite families having at most two points outside the 4-GSP.
 $\\ \\$
Since $G$ cannot have more than 7 vertices by Lemma \ref{4csp}, we consider only $n=7$.
%
%
$\\ \\$
 Consider any visibility embedding $\xi$ of $G$.
Let $p_1$, $p_2$, $p_3$ and $p_4$ be collinear points representing a 4-GSP (denoted as $l$). 
If the remaining three points $p_5$, $p_6$ and $p_7$ form a 3-GSP disjoint from $l$,
then $G$ is not planar as it has $K_{3,3}$ as a subgraph. 
If  $p_5$, $p_6$ and $p_7$ are mutually visible, and they also see all points of $l$, then $G$ 
is not planar as it has $K_{3,3}$ as a subgraph. 
If $p_5$, $p_6$ and $p_7$ are on opposite sides of $l$, then, again $G$ is not planar as it has $K_{3,3}$
as a subgraph.
So, in every embedding, all points $p_5$, $p_6$ and $p_7$ are on the same side of $l$.
Therefore, an endpoint of every 3-GSP in $\xi$ is a point of $l$. We have the following lemma.
%
%
%
\begin{lemma}\label{app2}
If every visibility embedding of a planar point visibility graph $H$ has a 4-GSP but no 5-GSP,
then every visibility embedding of $H$ has at least three 3-GSPs edge disjoint from the 4-GSP.
\end{lemma}
\begin{proof} 
Since $H$ has at most $15$ edges due to Euler's condition, and a complete graph on seven vertices 
has $21$ edges, there are at least six invisible pairs in $H$.
Three of these invisible pairs correspond to the 4-GSP. So, the remaining three invisible pairs must 
correspond to three 3-GSPs edge disjoint-from the 4-GSP.
%
%
  \end{proof}
%
\noindent Due to the above Lemma, we must ensure that three new 3-GSPs are formed in $\xi$, by adding $p_5$, $p_6$ and $p_7$.
We add $p_5$ and $p_6$ to construct the first new 3-GSP as shown in Fig. \ref{new13}, excluding symmetric cases. 
Then $p_7$ is added to these embeddings forming two more 3-GSPs.
This can be realized only by placing $p_7$ at intersection points of pairs of lines containing exactly two points on each line.
$\\ \\$
Let us add $p_7$ to the embedding shown in Fig. \ref{new13}(a).
Removing symmetric embeddings, we have the 
following choices of positioning $p_7$ in the two new 3-GSPs: $\overline {p_2 p_7 p_6}$ and $\overline {p_3 p_7 p_5}$
(Fig. \ref{new14}(a)),
 $\overline {p_2 p_7 p_6}$ and $\overline {p_4 p_7 p_5}$
 (Fig. \ref{new14}(b)), 
  $\overline {p_3 p_7 p_6}$ and $\overline {p_4 p_7 p_5}$ (Fig. \ref{new14}(c)), 
  $\overline {p_2 p_5 p_7}$ and $\overline {p_3 p_6 p_7}$  (Fig. \ref{new14}(d)), 
  $\overline {p_2 p_5 p_7}$ and $\overline {p_4 p_6 p_7}$   (Fig. \ref{new14}(e)),
  $\overline {p_3 p_5 p_7}$ and $\overline {p_4 p_6 p_7}$   (Fig. \ref{new14}(f)). 
It can be seen that embeddings in Figs.  \ref{new14}(a), \ref{new14}(c), \ref{new14}(d) and  \ref{new14}(e) 
correspond to non-planar graphs, and 
embeddings in Figs.  \ref{new14}(b) and \ref{new14}(f)  
correspond to planar graphs isomorphic to each other.
Hence, only one particular planar
graph arises after adding $p_7$ to the visibility embedding in Fig. \ref{new13}(a).
\begin{figure}   
\begin{center}
\mbox{\begin{minipage} [b] {55mm}
\centerline{\hbox{\psfig{figure=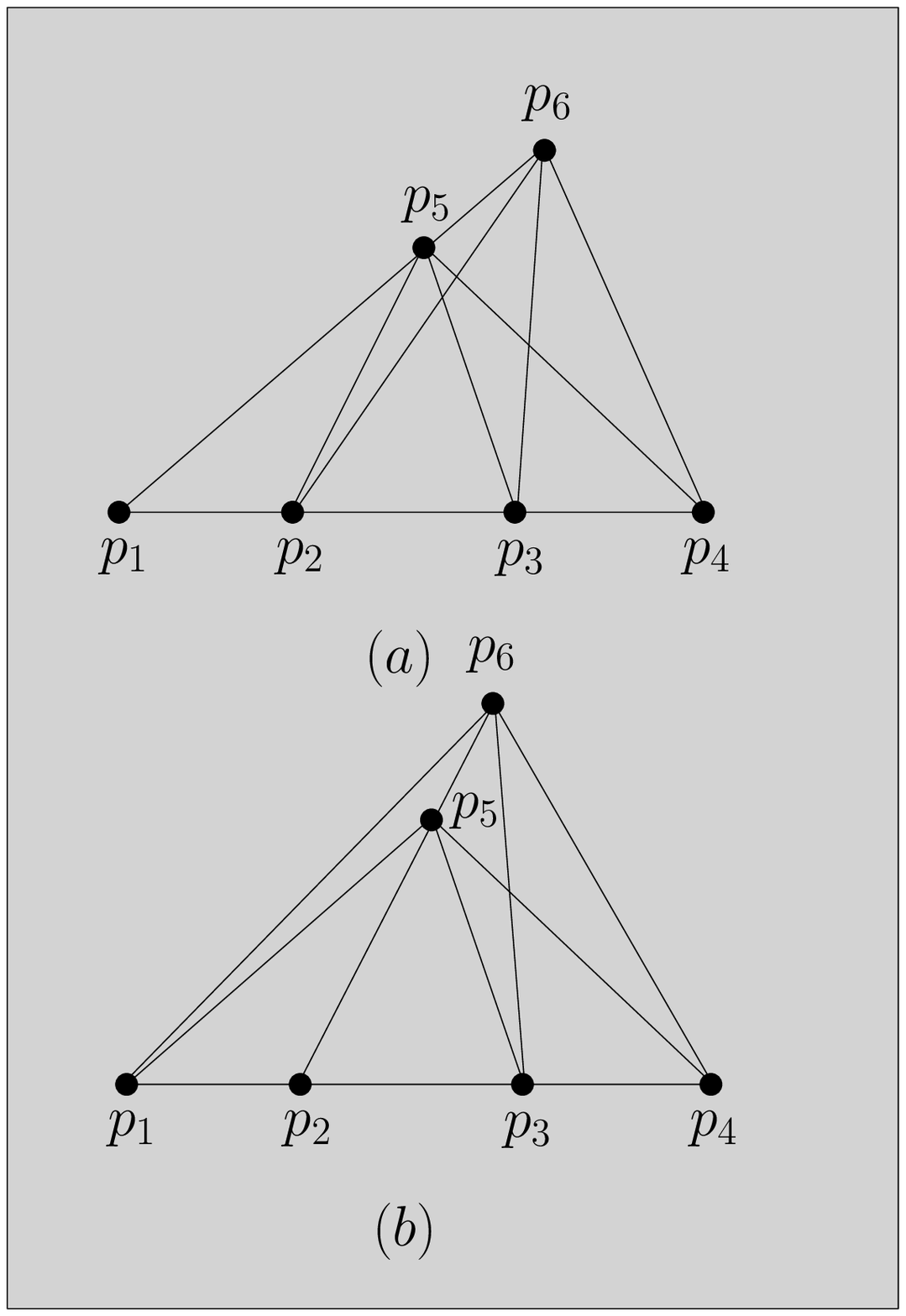,width=0.81\hsize}}}
\caption{Visibility embeddings of six points containing overlapping but edge disjoint 3-GSP and 4-GSP.}
 \label{new13}
\end{minipage}}\hspace{02mm}
\mbox{\begin{minipage} [b] {95mm}
\centerline{\hbox{\psfig{figure=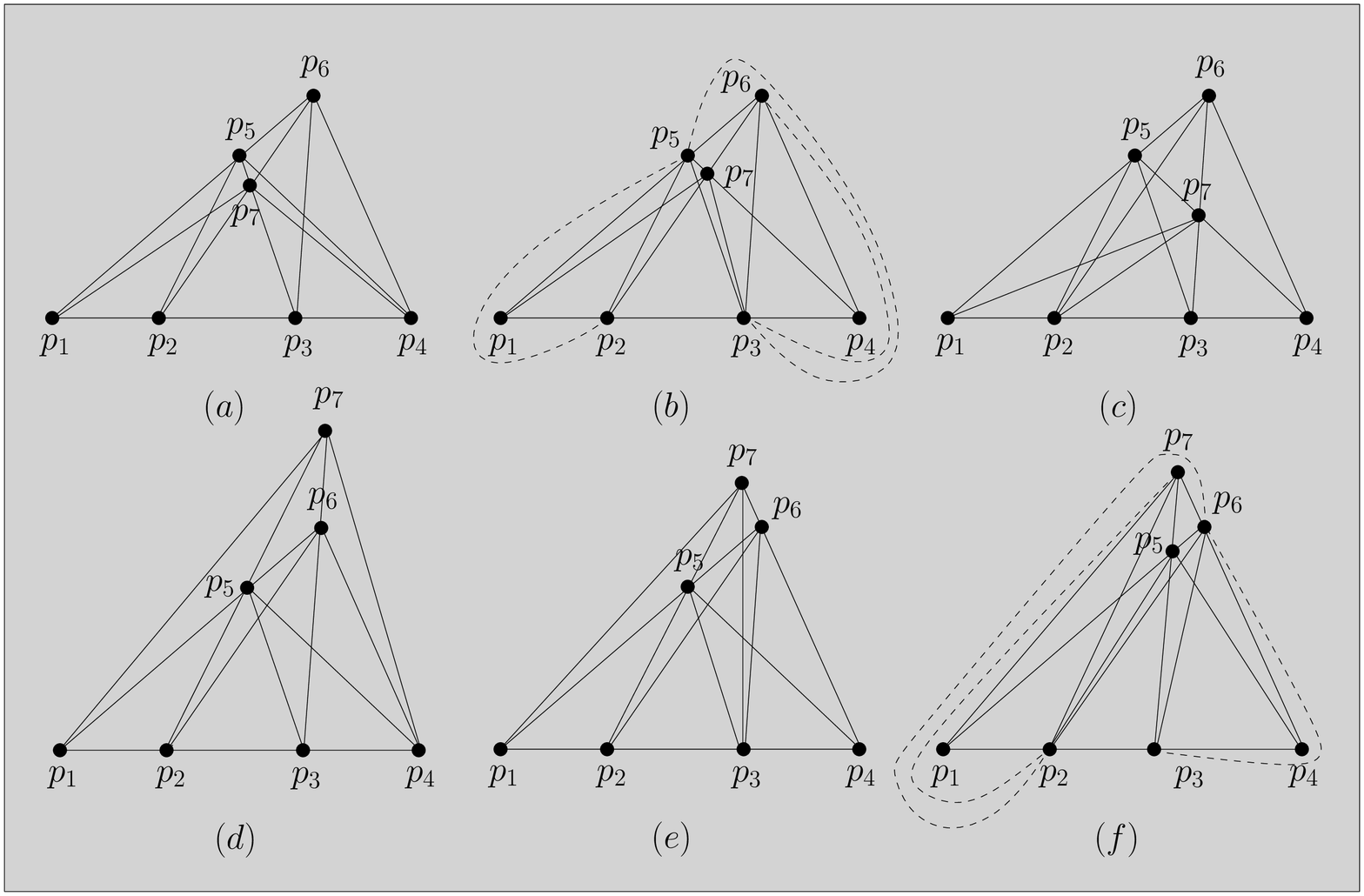,width=0.905\hsize}}}
\caption{Visibility embeddings of seven points after $p_7$ is added to the embedding in Fig. \ref{new13}(a).
Dotted lines show how the edge-crossings in the visibility embedding can be avoided in a planar embedding.}
 \label{new14}
\end{minipage}}
\end{center}
\end{figure}
\begin{figure}   
\begin{center}
\centerline{\hbox{\psfig{figure=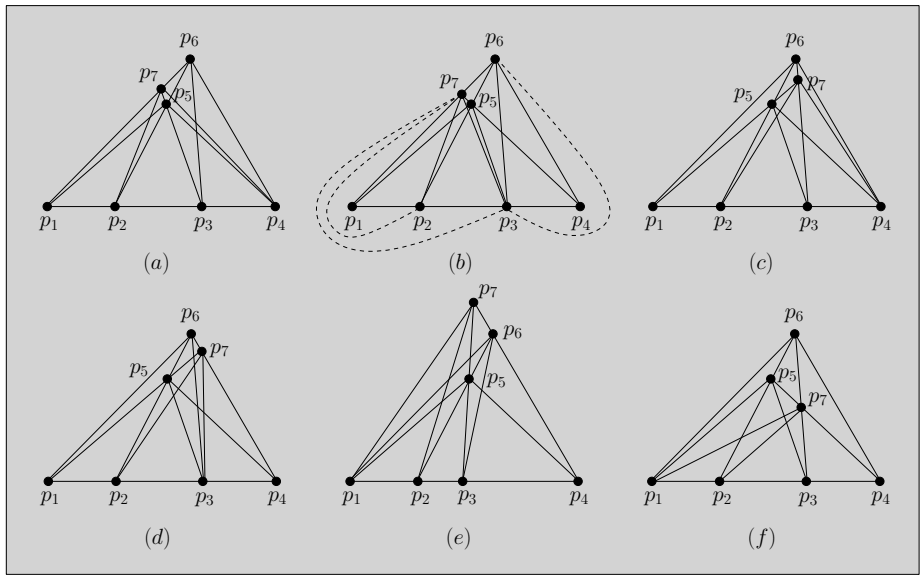,width=0.755\hsize}}}
\caption{Visibility embeddings of seven points after $p_7$ is added to the embedding in Fig. \ref{new13}(b).
Dotted lines show how the edge-crossings in the visibility embedding can be avoided in a planar embedding.}
 \label{new15}
\end{center}
\end{figure}
%
%
%
$\\ \\$
As before, let us add $p_7$ to the embedding shown in Fig. \ref{new13}(b). Removing symmetric embeddings, we have the 
following choices of positioning $p_7$ in the two new 3-GSPs: 
$\overline {p_1 p_7 p_6}$ and $\overline {p_3 p_5 p_7}$ (Fig. \ref{new15}(a)),
$\overline {p_1 p_7 p_6}$ and $\overline {p_4 p_5 p_7}$ (Fig. \ref{new15}(b)), 
$\overline {p_1 p_5 p_7}$ and $\overline {p_3 p_7 p_6}$  (Fig. \ref{new15}(c)),
$\overline {p_1 p_5 p_7}$ and $\overline {p_4 p_7 p_6}$ (Fig. \ref{new15}(d)),
$\overline {p_3 p_5 p_7}$ and $\overline {p_4 p_6 p_7}$ (Fig. \ref{new15}(e)), and
$\overline {p_3 p_7 p_6}$ and $\overline {p_4 p_7 p_5}$ (Fig. \ref{new15}(f)).
It can be seen that embeddings in Figs.  \ref{new15}(a), \ref{new15}(c), \ref{new15}(d),  \ref{new15}(e) 
and \ref{new15}(f) 
correspond to non-planar graphs, and the  
embedding in Fig.  \ref{new15}(b) 
corresponds to a planar graph.
$\\ \\$
\noindent But this embedding is already present in Fig. \ref{new14}.
So, no new planar graph arises after adding $p_7$ to the visibility embedding in Fig. \ref{new13}(b).
Thus,
one particular planar point-visibility graph of seven vertices is identified (see Fig. \ref{new14}(b)).

 \end{document}